\documentclass{article}
\usepackage[affil-it]{authblk}
\usepackage{enumitem}
\usepackage{graphicx}
\makeatletter
\def\ScaleIfNeeded{%
  \ifdim\Gin@nat@width>\linewidth
    \linewidth
  \else
    \Gin@nat@width
  \fi
} 
\makeatother
\let\Oldincludegraphics\includegraphics
{%
 \catcode`\@=11\relax%
 \gdef\includegraphics{\@ifnextchar[{\Oldincludegraphics}{\Oldincludegraphics[width=\ScaleIfNeeded]}}%
}%

\usepackage[margin=1in]{geometry}
\usepackage{setspace}
\usepackage{floatpag}
\floatpagestyle{plain}

\usepackage{algpseudocode}
\usepackage{algorithm}
\usepackage{bm}
\usepackage{natbib}
\bibpunct{[}{]}{,}{s}{}{;}

\usepackage{epsfig} 
\usepackage{amssymb} 
\usepackage{amsmath}
\usepackage{amsthm}
\usepackage{color}  
\usepackage[normalem]{ulem} 
\usepackage{verbatim}
\usepackage{pdflscape}
\usepackage{tabularx}
\usepackage{tabulary}
\usepackage{multirow}
\usepackage{xspace}
\usepackage{hyperref}
\usepackage{siunitx}

\usepackage{scalerel}
\usepackage[only,llbracket,rrbracket,llparenthesis,rrparenthesis]{stmaryrd}
\usepackage{accsupp} 

\usepackage{appendix}
\usepackage{float}

\newcommand{\segment}{segment\xspace}
\newcommand{\segments}{segments\xspace}

\newcommand{\name}{DyNAMiC Workbench\xspace}

\newcommand{\seg}[2]{\texttt{#1\textsuperscript{#2}}}

\newcommand{\cplx}[1]{\texttt{#1}}

\newcommand{\ratex}{\rho}
\newcommand{\rate}[1]{\ratex(#1)}

\newcommand{\fig}[1]{Fig.~#1}
\newcommand{\subfigure}[1]{\textbf{(#1)}}
\newcommand{\sect}[1]{Sec.~#1}

\newcommand{\alg}[1]{Alg.~#1}

\newcommand{\eqn}[1]{Eq.~#1}

\newcommand{\lem}[1]{Lemma~#1}

\newcommand{\term}[1]{``#1''}
\newcommand{\latin}[1]{\emph{#1}}

\newcommand*{\llbrace}{%
  \BeginAccSupp{method=hex,unicode,ActualText=2983}%
    \textnormal{\usefont{OMS}{lmr}{m}{n}\char102}%
    \mathchoice{\mkern-4.05mu}{\mkern-4.05mu}{\mkern-4.3mu}{\mkern-4.8mu}%
    \textnormal{\usefont{OMS}{lmr}{m}{n}\char106}%
  \EndAccSupp{}%
}
\newcommand*{\rrbrace}{%
  \BeginAccSupp{method=hex,unicode,ActualText=2984}%
    \textnormal{\usefont{OMS}{lmr}{m}{n}\char106}%
    \mathchoice{\mkern-4.05mu}{\mkern-4.05mu}{\mkern-4.3mu}{\mkern-4.8mu}%
    \textnormal{\usefont{OMS}{lmr}{m}{n}\char103}%
  \EndAccSupp{}%
}

\newcommand{\set}[1]{\left\{#1\right\}}
\newcommand{\mset}[1]{\llbrace #1 \rrbrace}
\newcommand{\seq}[1]{\langle #1 \rangle}
\newcommand{\li}[1]{\mathcal{#1}}
\newcommand{\type}[1]{\mathrm{#1}}
\newcommand{\mat}[1]{\mathbf{#1}}

\theoremstyle{plain}
\newtheorem{theorem}{Theorem}
\newtheorem{lemma}{Lemma}
\newtheorem{corollary}{Corollary}

\theoremstyle{definition}
\newtheorem{defn}{Definition}

\usepackage{times}

\usepackage[small,compact]{titlesec} 

\pdfoutput=1

\renewcommand{\segment}{domain\xspace}
\renewcommand{\segments}{domains\xspace}

\DeclareSIUnit\Molar{\text{M}}
\DeclareSIUnit\kcal{\text{kcal}}
\renewcommand{\Pr}[1]{\mathbb{P}\left[#1\right]}
\newcommand{\Ex}[1]{\langle #1 \rangle}
\renewcommand{\vec}[1]{\mathbf{#1}}

\title{A domain-level DNA strand displacement reaction enumerator allowing arbitrary non-pseudoknotted secondary structures\vspace*{-0.2cm}}
\author[1]{Casey Grun}
\affil[1]{Harvard University} 
\author[2]{Karthik Sarma} 
\author[2]{Brian Wolfe}
\author[3]{Seung Woo Shin}
\affil[3]{University of California, Berkeley}
\author[2]{Erik Winfree}
\affil[2]{California Institute of Technology}
\date{\today\vspace*{-0.5cm}}

\begin{document}
\maketitle
\section{Abstract}

DNA strand displacement systems have proven themselves to be fertile substrates for the design of programmable molecular machinery and circuitry.  
Domain-level reaction enumerators provide the foundations for molecular programming languages by formalizing DNA strand displacement mechanisms and modeling interactions at the \term{domain} level -- one level of abstraction above models that explicitly describe DNA strand sequences.  Unfortunately, the most-developed models currently only treat pseudo-linear DNA structures, while many systems being experimentally and theoretically pursued exploit a much broader range of secondary structure configurations.  Here, we describe a new domain-level reaction enumerator that can handle arbitrary non-pseudoknotted secondary structures and reaction mechanisms including association and dissociation, 3-way and 4-way branch migration, and direct as well as remote toehold activation.  To avoid polymerization that is inherent when considering general structures, we employ a time-scale separation technique that holds in the limit of low concentrations.  This also allows us to ``condense'' the detailed reactions by eliminating fast transients, with provable guarantees of correctness for the set of reactions and their kinetics.  We hope that the new reaction enumerator will be used in new molecular programming languages, compilers, and tools for analysis and verification that treat a wider variety of mechanisms of interest to experimental and theoretical work. We have implemented this enumerator in Python, and it is included in the DyNAMiC Workbench Integrated Development Environment.

\section{Introduction}\label{sec:introduction}

\begin{figure}[!htb]
	\includegraphics{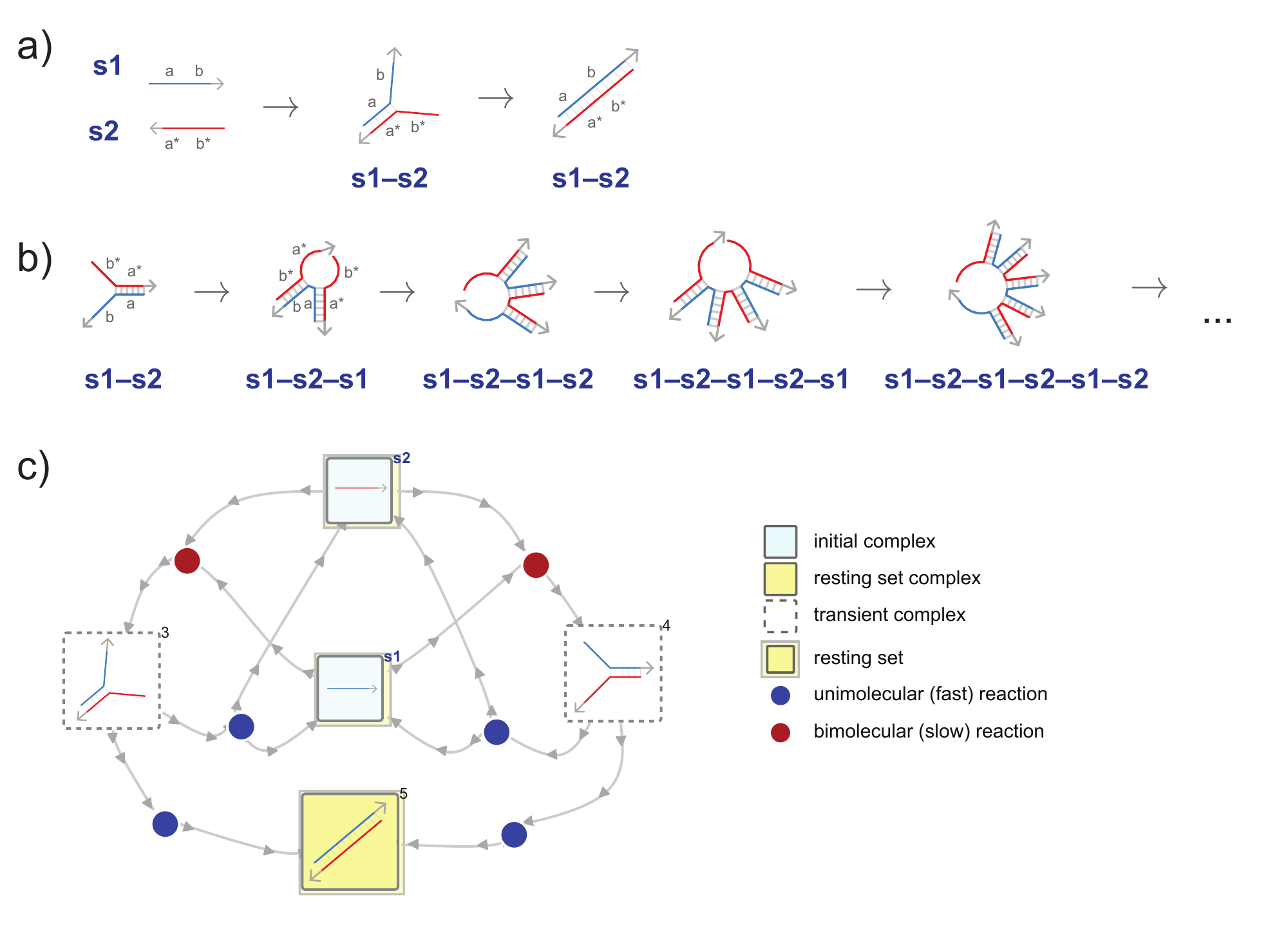} \vspace{-7mm}
	\caption[Polymerization]{The Polymerization Problem. Reaction enumerators must avoid generation of infinitely-long polymers.  
	\subfigure{a} Intended behavior for a simple system of two strands. Binding between the red strand and the blue strand results in the formation of a stable duplex (the binding could be nucleated by binding between either \segment \seg{a}{}/\seg{a}{*} or \segment \seg{b}{}/\seg{b}{*}---here we show only the former for simplicity).
	\subfigure{b} Pathological behavior in an enumerator without a separation of timescales. Repeated bimolecular association reactions are allowed to occur before the intramolecular binding reaction, generating implausibly long polymers. Depending on the order in which the enumerator searches for reactions, this could prevent such an enumerator from \emph{ever} finding the kinetically- and thermodynamically-favored complex---the simple duplex.
	\subfigure{c} Reaction network generated by our enumerator. Separation of reactions into \term{fast} (unimolecular) and \term{slow} (bimolecular) allows the classification of certain complexes (dashed boxes) as \term{transients}; by prohibiting the transient complexes from participating in bimolecular reactions, the pathway generated by our enumerator converges on the intended duplex.}
	\label{fig:polymer}
\end{figure}

A series of inspiring demonstrations during the last 15 years have shown DNA to be a robust and versatile substrate for nanoscale construction and computation\cite{Zhang:2011hw,Krishnan:2011ks,Seeman:2010fa,Chen:2010un}. DNA has been used to build tile assemblies\cite{Winfree:1998fs,Rothemund:2004dm}; arbitrary shapes\cite{Rothemund:2006hb,Wei:2012ip}; molecular motors\cite{Yurke:2000ex,Turberfield:2003uo,Venkataraman:2007bw}, walkers\cite{Shin:2004gn,Sherman:2004tv,Yin:2008gb,Omabegho:2009dr}, and robots\cite{Ding:2006hn,Muscat:2011he}; analog circuits that perform amplification of concentrations\cite{Zhang:2007gt,Yin:2008gb,Li:2012gx} and polymer lengths\cite{Dirks:2004cb,Yin:2008gb}; synthetic transcriptional circuits\cite{Kim:2011synthetic}; molecular logic gates\cite{Seelig:2006hw}; Turing-universal stack machines\cite{Qian:2011to}; and application-specific analog and digital circuits that perform non-trivial calculations\cite{Qian:2011hu,Qian:2011gh,Chen:2013kf,Zhang:2011hw}.
A common abstraction is to describe these complex systems in terms of a number of \term{\segments}---contiguous sequences of nucleotides that are intended to participate in hybridization as a group; complementary \segments are intended to interact, and all other \segments are not. Once a system has been described in terms of \segments, a computational sequence design package can generate a sequence of nucleotides for each \segment in the system, implementing the desired complementarity rules\cite{Zadeh:2011ku,Zhang:2011uc}. 
Developing such increasingly complex dynamic DNA nanosystems requires tools both for the automated design of these systems, but also for software-assisted verification of their behavior\cite{Phillips:2009ec,Grun:2014vd}.

One mode of verification is to check whether a set of DNA complexes can react to produce the desired product species (or various anticipated intermediate species). Given a domain-level model, this can be done by enumerating the network of all possible reactions that can occur, starting from a finite set of initial complexes. The result is a chemical reaction network (equivalently a Petri net), connecting these initial complexes to some set of intermediate complexes. For the sake of this paper, we use the term \term{enumeration} to refer to the process of generating a chemical reaction network, given a set of initial complexes and a set of rules for their possible interactions.

This enumeration process, by itself, does not fully describe the expected kinetic behavior of the system, because the concentration (or number of copies) of each initial species is not yet specified.
However, the
enumerated network of reactions forms the basis of a system's dynamics, and therefore a great deal can often be observed about a system's behavior simply by examining the network of reactions. Unintended side reactions that could hamper the kinetics of the system or cause unanticipated failure modes can be identified.  Furthermore, given initial concentrations (or counts), the chemical reaction network can be simulated using deterministic ordinary differential equations (ODEs) or stochastic methods, and in some cases global properties of the dynamics can be analyzed.
In dynamic DNA nanotechnology, this type of enumeration, like the design, has often been performed by hand; this is clearly infeasible for all but the most simple systems, as the number of possible unintended interactions between two species grows quadratically in the number of \segments.  Worse, the combinatorial structure inherent in the rearrangement and assembly of strands can give rise to an exponential---potentially even infinite---number of species.  Thus, automated methods for performing the enumeration are essential, and for tractability it is necessary that they focus attention on only the most relevant possible complexes and reactions.
Rule-based models developed for concisely representing combinatorial structures in systems biology, such as BioNetGen\cite{faeder2009rule} or Kappa\cite{feret2009internal}, could in principle be used for DNA systems, but appropriate rules for the relevant DNA interactions would have to be provided by the user for each system.

Several previous efforts have explored the \latin{in silico} testing and verification of dynamic DNA nanotechnology systems using a built-in set of rules. 
Most notably, Phillips and Cardelli have demonstrated the DNA circuit compiler and analysis tool ``Visual DSD,'' which was subsequently extended by Lakin et al.\cite{Phillips:2009ec,Lakin:2011vo,Lakin:2012hb} The analysis component of DSD allows for enumeration of possible reactions between a set of initial complexes, and then simulation of those reactions---either using a system of ODEs or a stochastic simulation. DSD also allows a hybrid ``just-in-time'' enumeration model which interleaves enumeration with simulation; in this model, reactions are stochastically selected for pursuit by the enumerator based on their estimated probabilities (rather than the entire network being enumerated exhaustively) . However, DSD has a limited system of representation for strand displacement systems; in particular, the formalism cannot represent branched junction, hairpin, or multiloop structures. Similarly, the DSD model cannot currently express reactions relying on looped or branched intermediates, such as 4-way branch migration or remote toehold-mediated 3-way branch migration\cite{Genot:2011gt}. These structural motifs and reaction mechanisms have proven useful in recent experimental work for self-assembly\cite{Yin:2008gb,Dabby:2013uq}, locomotion\cite{Venkataraman:2007bw,Muscat:2011he}, imaging\cite{Choi:2010bm}, and computation. Earlier work by Nishikawa et al. included a similar joint enumeration and simulation model that allows arbitrary secondary structures using strands with ``abstract bases'' analogous to our ``\segments''; combinatorial explosion was controlled by only allowing interactions between complexes above some threshold concentration\cite{Nishikawa:2001cp}. More recent work by Kawamata et al. uses a graph-based model to represent and simulate reactions between species\cite{Kawamata:2011wv,Kawamata:2012abstraction}.

Here, our goal was to develop a \segment-level reaction enumerator that (a) separates the enumeration and simulation steps so that the resulting reaction network can be rigorously analyzed, (b) models a wide range secondary structures and conformational mechanisms used in dynamic DNA nanotechnology, and (c) controls combinatorial explosion using an approximation that is valid in a well-defined limit, thus allowing simplification of the resulting network.
Our enumerator, which we call Peppercorn, exhaustively determines the possible hybridization reactions between a set of initial complexes, as well as between any complexes generated as products of these reactions, yielding a complete reaction network.  The enumerator can represent arbitrary non-pseudoknotted complexes, and handles 3-way and 4-way branch migration, including reactions initiated by remote toeholds.
It enforces a separation of timescales in order to avoid implausible polymerization that would otherwise result from this degree of generality.
Finally, it includes a scheme for \term{condensing} reactions to consider only slow transitions between groups of species---excluding transient intermediate complexes. We have implemented the Peppercorn enumerator in Python, and it has been included as part of the DyNAMiC Workbench Integrated Development Environment\cite{Grun:2014vd}, which is available online with a graphical user interface (\url{www.molsys.net/workbench}).

The driving issue in the development of our enumerator has been the need to model as wide a range of secondary structures as possible.  We chose non-pseudoknotted secondary structures, a class that includes not only the linear and ``hairy linear'' structures of DSD, but also branched tree-like structures with hairpin loops, bulge loops, interior loops, and multiloop junctions (\fig{\ref{fig:structures}}).   (Pseudoknotted structures, in which branches of a tree can contact each other to form cycles, require more complicated energy models, have geometrical constraints, and are not extensively used in dynamic DNA nanotechnology yet, so version 1.0 of Peppercorn does not allow them.  See Appendix D and \fig{\ref{fig:pseudoknots}} for further explanation.)

The choice to allow arbitrary non-pseudoknotted structures, and arbitrary hybridization interactions between them,  introduces some problems that don't arise in restricted models such as Visual DSD.  
For example, some sets of species may allow for the generation of infinitely long polymers (\fig{\ref{fig:polymer}}); a naive attempt to enumerate all possible reactions in such a system would prevent the algorithm from terminating. In the laboratory, these polymers will not occur at low concentrations if kinetically faster reactions are possible that preclude polymerization. Therefore, the enumerator must provide \emph{some} plausible approximation for a separation of kinetic timescales.  We find here that the low-concentration limit, in which unimolecular reactions are sufficiently faster than bimolecular reactions, provides a clean and rigorous basis for a semantics that avoids the spurious polymerization issue.

Finally, it is desirable that the results of this enumeration be \emph{interpretable} to the user---an excessively complex reaction network will be useful only for performing kinetic simulations and will make it hard to distinguish intended reaction pathways from unintended side reactions. Simplification may be necessary for further verification (for instance, comparison of the system's predicted behavior to a specification, given in a higher level language). It is essential that any simplification has some guaranteed correspondence to the full, detailed reaction network. Given a presumed separation of timescales, one way to make the network more interpretable is to show only the ``slow'' reactions, and to assume that ``fast'' reactions happen instantaneously. We will explore this idea in detail later.

As a motivating example, we show the 3-arm junction system of Yin et al. in \fig{\ref{fig:catalytic-3arm}}\cite{Yin:2008gb}; this system involves the assembly of a 3-arm junction from a set of metastable hairpins, triggered by the addition of a catalyst (\fig{\ref{fig:catalytic-3arm}a}). The catalyst opens one hairpin by toehold-mediated branch migration, exposing another toehold (previously sequestered within the hairpin stem); this toehold similarly opens the second hairpin, which can open the third hairpin. In the final step, a four-strand intermediate complex (containing the three hairpins and the initiator) can collapse to the final 3-arm junction structure, releasing the single-strand catalyst. Even for this simple process, the full, detailed reaction network is rather complicated (because of the separate toehold-nucleation steps necessary for attachment of each hairpin to the growing structure, as well as various unintended side reactions). These steps result in transient intermediates that quickly decay to a more stable \term{resting complex.} By skipping these transient intermediates (and depicting only bimolecular transitions between sets of resting complexes), the reaction network can be greatly simplified (\fig{\ref{fig:catalytic-3arm}b}).

\begin{figure}
	\includegraphics{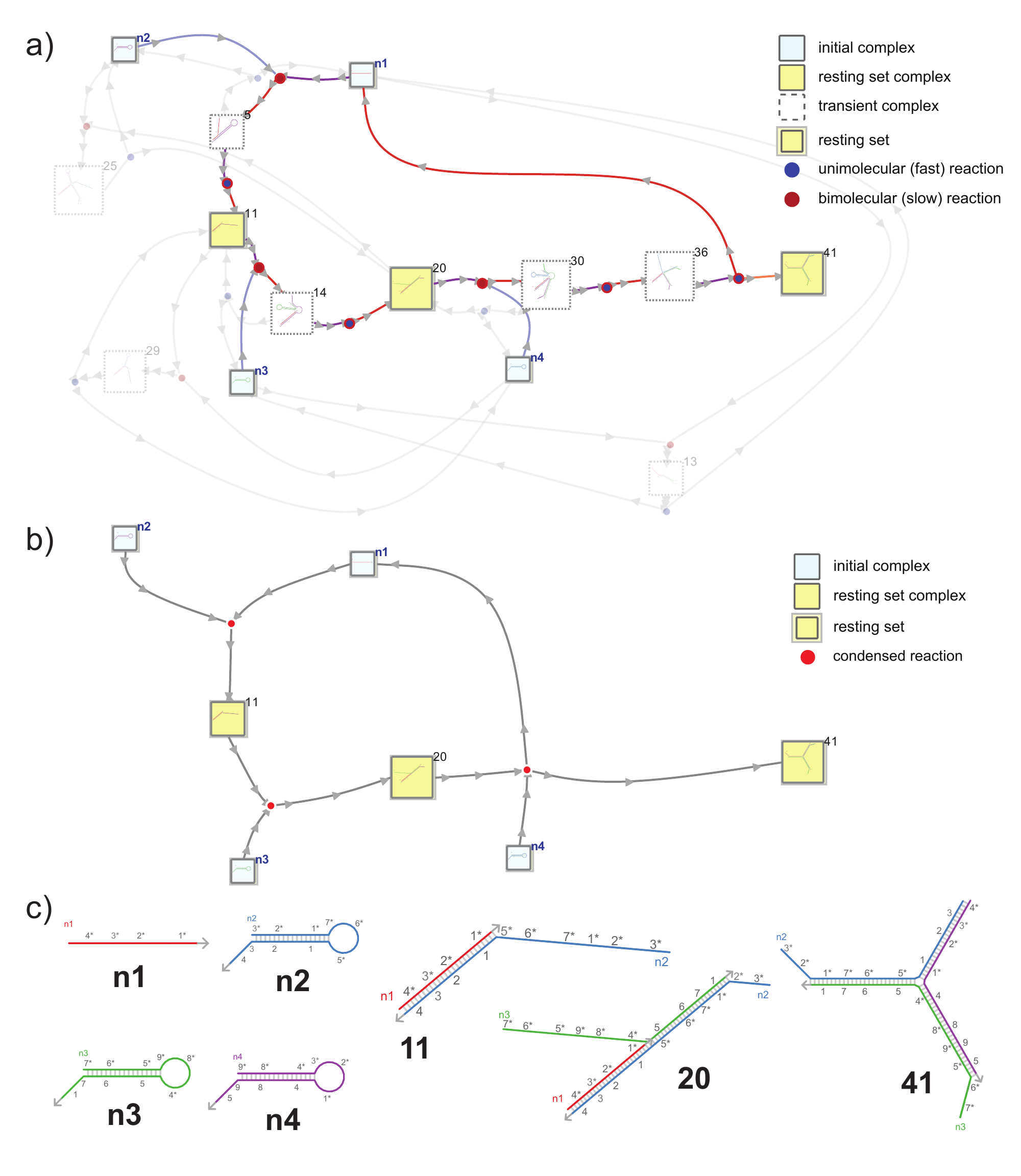} \vspace{-10mm}
	\caption[Catalytic 3-arm junction]{Catalytic 3-arm junction formation. We show the catalytic 3-arm junction assembly process (originally demonstrated by Yin et al. \cite{Yin:2008gb}), as enumerated by our software; the intended pathway is highlighted in red and purple. Boxes with rounded edges represent complexes---initial complexes have blue backgrounds, \term{transient} complexes have dashed borders (e.g. complex \cplx{5}), and \term{resting complexes} have yellow backgrounds. Circular nodes are reactions---\term{unimolecular} (blue) and \term{bimolecular} (red) reactions. \term{Resting sets} are shown as light yellow boxes around complexes (in this case, each resting set has only one complex). 
	\subfigure{a} Detailed reactions. A set of three strands (\cplx{n2}, \cplx{n3}, \cplx{n4}) begin as metastable hairpin monomers. A single-stranded initiator (\cplx{n1}) induces toehold attachment (\cplx{5}) and opening of the first hairpin, exposing a sequestered toehold (\cplx{11}). The opened hairpin-initiator complex can bind the second hairpin (\cplx{n3}) in a similar process, and so on for the third hairpin (\cplx{n4}). Finally, remote toehold-mediated branch migration allows \cplx{n4} to displace the initiator \cplx{n1}, allowing for multiple turnover.
	Greyed-out reactions and complexes indicate non-productive binding.
	\subfigure{b} Condensed reaction graph, showing only reactions between \term{resting sets.} 
	\subfigure{c} Key complexes are shown in detail. }
	\label{fig:catalytic-3arm}
\end{figure}


\section{Reaction Enumeration Model}


\subsection{Primitives and Definitions}\label{sec:primitives-and-definitions}

We begin by defining some terms necessary for a discussion of the reaction enumeration algorithm and process. These terms are also demonstrated in \fig{\ref{fig:catalytic-3arm}}.

\begin{defn}
	A \emph{\segment} $d = (z, \sigma)$ is a tuple where $z$ is the name of the
	\segment and $\sigma$ is the (optional) string of nucleotides from the alphabet
	$\Sigma = \set{A, C, T, G}$. The length of a \segment $d$ is written as
	$|d|$ and represents the number of nucleotides in the sequence $\sigma$. 
	The enumerator semantics depend only on $|d|$ and not on the actual sequence,
	thus allowing the domain-level dynamics to be specified prior to sequence design.  A
	\segment $d = (z, \sigma)$ has a \emph{complementary} \segment $d^*$ whose name
	is generally $z^*$ and whose sequence is the Watson-Crick complement
	of $s$. 
\end{defn}
\begin{defn}
	A \emph{strand} $s = (z, D)$ is a tuple where $z$ is the name of the
	strand and $D$ is a sequence of \segments $\seq{d_1, d_2, \ldots}$, ordered from 
	the 5' end of the strand to the 3' end. Two strands are equal if and only if their 
	names are equal and they contain identical sequences of \segments. 
\end{defn}
\begin{defn}
	A \emph{structure} $T = \seq{T_1, T_2, \ldots T_n}$ (for a sequence of
	$n$ strands) is a sequence of lists of bindings, where each binding is either a pair of integers or $\varnothing$. That is,
	$T_i = \seq{T_{i,1}, T_{i,2}, \ldots}$ is the structure for strand
	$i$. $T_{i,j} = (s_{i,j}, d_{i,j})$ is a tuple indicating that \segment
	number \emph{j} within strand number \emph{i} is bound to \segment number $d_{i,j}$ in strand
	number $s_{i,j}$. This encoding is redundant because
	$T_{i,j} = (i', j') \iff T_{i', j'} = (i, j)$. If a \segment is unbound,
	$T_{i,j} = \varnothing$.
\end{defn}
\begin{defn}
	A \emph{complex} $c = (z, S, T)$ is a tuple where $z$ is the name, $S$
	is the sequence of strands in the complex, and $T$ is the structure of
	the complex. 

	\begin{itemize}
	
	\item
		A complex can be \emph{rotated} by circularly permuting the elements
		of $S$ and $T$ (e.g. $S' = \seq{S_n, S_1, S_2, \ldots, S_{n-1}}$)
		and setting $s_{i,j} = 1 + ((s_{i,j} + n - 1) \mod n)$ for all $s_{i,j} \in T_i$ for all $T_i \in T$). Note that for a non-pseudoknotted complex, the sequence of strands has a unique circular ordering in which the base pairs are properly nested, but there are multiple equivalent circular permutations of the complex\cite{Bois:2007vp}. Therefore we say the
		resulting complex is in \emph{canonical form} if the
		lexicographically (by the strand name) lowest strand $S_\ell$ appears first (e.g.
		$S_1 = S_\ell$). A complex may be rotated repeatedly in order to
		achieve canonical form. 
		Two complexes in canonical form are equal if they contain the same
		strands in the same order, and their structures are equal.
	\item
		A complex must be \term{connected}, which means that all strands in the 
		complex are connected to one another. Two strands $S_1$ and $S_2$ are \term{connected} 
		if either $S_1$ is bound to $S_2$ or $S_1$ is bound to some other strand 
		that is connected to $S_2$. 
	\end{itemize}
\end{defn}
\begin{defn}
	A \emph{reaction} $r = (A, B)$ is a tuple where $A$ is the multiset of
	reactants and $B$ is the multiset of products, where reactants and products are both complexes. 
	We write $\rate{r}$ to represent the rate of $r$.

	\begin{itemize}
	\item
		The \emph{arity} $\alpha(r)$ of a reaction $r$ is a pair
		$(|A|, |B|)$, where $|A|$ counts the total number of elements in $A$, 
		and similarly for $|B|$. Any reaction with arity $(1, n)$ is
		\emph{unimolecular}; reactions with arity $(2, n)$ are
		\emph{bimolecular}, and those with other arities are \emph{higher
		order}. For the sake of this paper, we do not consider $(0, n)$ or $(n,0)$ reactions
		(those which birth new products from no reactants or cause all reactants to disappear).
	\item
		A reaction may be classified as \emph{fast} or \emph{slow}; we make
		this separation based on the arity of the reaction (unimolecular
		reactions are fast, while bimolecular and higher-order reactions are slow). This
		assumption is rooted in the Law of Mass Action for chemical kinetics.  
		For a unimolecular reaction $A \to B$,
		the rate $\ratex_\text{uni}$ of the reaction is given by $\ratex_\text{uni} = k_\text{uni} [A]$
		where $[A]$ is the concentration of species $A$, and $k_\text{uni}$ is a rate constant.
		For a bimolecular reaction $A + B \to C$, the rate $\ratex_\text{bi}$ of the reaction is given by 
		$\ratex_\text{bi} = k_\text{bi} [A][B]$, where $[A]$ is the concentration of species $A$, $[B]$ is the concentration of $B$, 
		and $k_\text{bi}$ is a rate constant.
		
		This means that, as the concentration of all species decreases, the rates of 
		bimolecular reactions decreases more quickly than the rates of unimolecular reactions.
		Intuitively, this is because bimolecular reactions require the increasingly improbable 
		collision of two species (rather than one species simply achieving sufficient transition
		energy). 
		Therefore, the assumption of a ``clean'' separation of timescales (such that \emph{all} unimolecular reactions
		occur at a faster rate than any bimolecular reactions) is only valid at sufficiently low concentrations.

	\item
		For a set $R$  of reactions, we sometimes write $R_f$ to represent the fast reactions and $R_s$ to
		represent the slow reactions, such that $R = R_f \cup R_s$. Finally, it will be sometimes useful to 
		partition  $R_f$ into $(1,1)$ and $(1,n)$ reactions, such that $R_f = R_f^{(1,1)} \cup R_f^{(1,n)}$, 
		where by convention $n>1$.
	\end{itemize}
\end{defn}
\begin{defn}
	A \emph{resting set} $Q$ is a set of one or more
	complexes---strongly-connected by fast (1,1) reactions---with no
	\emph{outgoing} fast reactions of any arity. That is, fast reactions can
	interconvert between any of the complexes within $Q$, but no fast
	reaction is possible that transforms a complex in $Q$ to any complex
	outside of $Q$. We write this as:
	$\forall r = (A, B) \in R_f : A \subseteq Q \rightarrow B \subseteq Q$.

	\begin{itemize}
	
	\item
		Complexes which belong to resting sets are called \emph{resting complexes}; all other complexes are
		\emph{transient complexes}.
	\end{itemize}
\end{defn}
\begin{defn}
	A \emph{reaction network} is a pair $G = (C, R)$ where $C$ is a set of
	species (either complexes or resting sets) and $R$ is 
	a set of reactions between those species.

	\begin{itemize}
	\item
		When $C$ is a set of complexes, we refer to $G$ as a \term{detailed reaction network.} In \sect{\ref{sec:condensing-reactions}} we discuss 
		reaction networks where $C$ is a set of resting sets; these are \term{condensed reaction networks.} 
	\item
		Although we may informally refer to a reaction network as a ``reaction graph,'' 
		the presence of bimolecular reactions means that most reaction networks 
		\emph{cannot} be simply represented as a graph whose vertices are complexes and
		whose reactions are edges. 

		If $R$ contains only $(1,1)$ reactions, than we can consider $G$ to
		be a directed graph, where $C$ are the vertices, and each reaction
		$r = (\set{a}, \set{b}) \in R$ is an edge connecting the reactant
		complex $a$ to the product complex $b$. We use this representation 
		to define a notion of resting sets (above), and to provide a scheme for 
		condensing a reaction network.

		Alternatively, a reaction network may be represented as bipartite graph $G' = (V, E)$, where the 
		set of vertices $V = C \cup R$, and the edges in $E$ are pairs $(c, r)$ or $(r, c)$ for some $r \in R$, 
		$c \in C$. For instance, this would allow a reaction $r$ such as:
		$$X + Y \to Z$$
		to be represented by the graph $G = (\set{X, Y, Z} \cup \set{r}, \set{(X, r), (Y, r), (r, Z)})$. We use this 
		representation only when visually depicting reaction networks, as in \fig{\ref{fig:catalytic-3arm}}. 
		In all other cases in this paper, the graphs we refer to will be unipartite where the vertices are species.
	\end{itemize}
\end{defn}
\begin{defn}
	A \term{reaction enumerator} is therefore a function which maps a set of complexes $\li{C}$ to a set of reactions $\li{R}$ between those complexes. Therefore application of a reaction enumerator to a set of complexes yields a reaction network.
\end{defn}

We note briefly that these notions are related to other models of parallel computing, particularly Petri nets \cite{Cook:2009kp,Goss:1998tf}. Our reaction networks could also be considered Petri nets, where the species (either complexes---for a detailed reaction network, or resting sets---for a condensed reaction network) are \term{places} and the reactions are \term{transitions.} In a test tube, there would be a finite number of molecules for each individual species---in the language of Petri nets, molecules are called \term{tokens}, and a \term{marking} is an assignment of token counts (molecule numbers) to each place (species).

We emphasize that the ``enumeration'' task at hand is determination of the CRN/net itself---we are \emph{not} enumerating possible molecule counts/markings of the net. Rather, we assume a finite initial set of complexes and---given some rules about how those complexes may interact---attempt to list all possible reactions that may occur. In the next section, we describe these interaction rules.


\subsection{Reaction types}\label{sec:reaction-types}

The enumerator relies on the definition of several \term{move functions}---a move function takes a complex or set of complexes and
generates all possible reactions of a given type. In principle, many reaction types are possible; however, we restrict ourselves to the following four basic types, further
classifying each type by arity. Here we provide the move functions for each of these reaction types. The different reaction types are summarized in \fig{\ref{fig:rxn-types}}.

Note: when clear from context which complex is under consideration, we will sometimes write $d_{i,j}$ to refer to \segment $d_j$ on strand $s_i$ in this complex. Similarly, if $X = (i,j)$, we may write $d_X$ to refer to $d_{i,j}$. 

\begin{description}

\item[Binding] -- two complementary, unpaired \segments hybridize to form a new
	complex (\fig{\ref{fig:rxn-types}a}).

	\begin{itemize}
	
	\item
		(1,1) binding -- binding occurs between two \segments in the same
		complex. These reactions can be discovered by traversing each \segment on each strand and looking ahead for higher-indexed \segments that are complementary (within a multiloop, skipping over stems that lead to branched structures, in order to avoid generating pseudoknots); a separate reaction is generated for each separate complementary pair. See \alg{\ref{alg:bind11}} for pseudocode.

	\item
		(2,1) binding -- binding occurs between two \segments on different
		complexes. These reactions can be found by comparing each \segment in one complex $c$ to each \segment in another complex $c'$, generating a separate reaction to join each complementary pair that is found. See \alg{\ref{alg:bind21}} in the apppendix for pseudocode.

	\end{itemize}
\item[Opening] -- two paired \segments in a complex detach (\fig{\ref{fig:rxn-types}b}). We assume that there is some threshold length $L$, such that duplexes less than or equal to $L$ nucleotides long can detach rapidly (in unimolecular fast reactions), while \segments longer than $L$ bind irreversibly. $L$ is a parameter that can be set by the user; the default value is $L = 8$.

	These reactions can be found by examining each \segment $d_j$ on a complex that is bound to a higher-indexed \segment, looking to the left and right of that \segment (along the same strand) to determine the total length $\ell$ in nucleotides of the helix containing $d_j$. If $\ell < L$, a new reaction is generated that detaches \emph{all} bound \segments in the helix. See \alg{\ref{alg:open}} in the appendix for pseudocode.

	\begin{itemize}
	
	\item
		(1,1) opening -- two paired \segments in a complex detach, but
		the complex remains connected

	\item
		(1,2) opening -- two paired \segments in a complex detach, and
		the complex dissociates into two complexes
	\end{itemize}
\item[3-way] -- an unpaired \segment replaces one \segment in a nearby duplex
	(branch migration) (\fig{\ref{fig:rxn-types}c}). These reactions can be discovered as follows: for each bound \segment with an adjacent unbound \segment on each strand (the invading \segment), look first to the left and then to the right of that \segment---skipping over internal loops---for a third bound \segment (the displaced \segment) that is complementary to the invading \segment. If the displaced \segment is directly adjacent to the bound \segment (on the same strand), then this is \term{\emph{direct} toehold-mediated branch migration}; otherwise (e.g. if the bound \segment and the displaced \segment are separated by an internal loop), this is known as \term{\emph{remote} toehold-mediated branch migration.}\cite{Genot:2011gt}\footnote{An example can be seen in \fig{\ref{fig:catalytic-3arm}a}, in the reaction $\cplx{36} \to \cplx{41} + \cplx{n1}$} See \alg{\ref{alg:3way}} in the appendix for pseudocode.

	\begin{itemize}
	
	\item
		(1,1) 3-way -- the branch migration results in a complex which
		remains connected
	\item
		(1,2) 3-way -- the branch migration releases a complex
	\end{itemize}
\item[4-way] -- a four-arm junction is re-arranged such that hybridization is
	exchanged between several strands (\fig{\ref{fig:rxn-types}d}). See \alg{\ref{alg:4way}} in the appendix for pseudocode.

	\begin{itemize}
	\item
		(1,1) 4-way -- the branch migration results in a complex which
		remains connected
	\item
		(1,2) 4-way -- the branch migration releases a complex
	\end{itemize}
\end{description}

\begin{figure}[t]
\centering
\includegraphics[width=\textwidth]{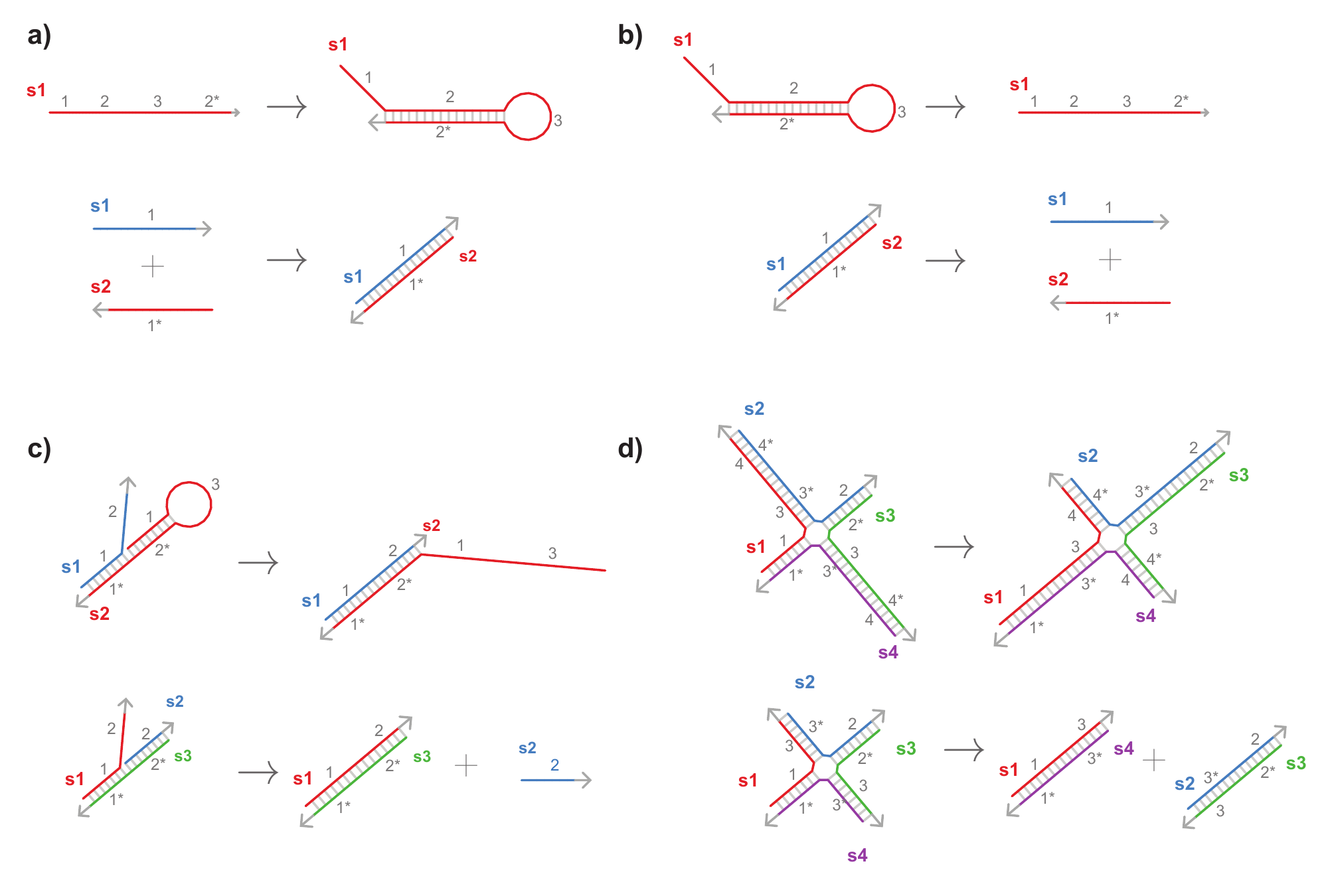} \vspace{-12mm}
\caption{Available reaction types. 
\subfigure{a} Bind (1,1) and (2,1). 
\subfigure{b} Open (1,1) and (1,2). 
\subfigure{c} 3-way branch migration (1,1) and (1,2). 
\subfigure{d} 4-way branch migration (1,1) and (1,2)}
\label{fig:rxn-types}
\end{figure}


\subsection{Separation of timescales}\label{sec:fast-and-slow-reactions}

Reactions are classified as \emph{fast} or \emph{slow} according to
their arity; as discussed above ($1,n$) reactions are fast for all $n\ge0$, while all other reactions
are slow. Complexes are partitioned into
\term{transient complexes} and \term{resting complexes.} Resting complexes are those complexes that are part of a \term{resting set}, while transient complexes are all other complexes. resting sets are calculated as follows: consider a reaction network $C, R$; we can construct a graph $G = (C, R_f^{(1,1)})$ where $R_f^{(1,1)}$ is the set of fast $(1,1)$ reactions in the reaction network. We compute the set $W$ of \term{strongly-connected components} of $G$; Each component $W_i \in W$ is a resting set if and only if there are no outgoing fast reactions for any complex in $W_i$. That is, for some resting set $Q \in W$, any fast reaction consuming a complex in $Q$ produces only products that are in $Q$. 

Intuitively, resting sets are sets of complexes that can interconvert rapidly between one another, whereas transient complexes are those which (via only fast reactions) rapidly react to become different complexes (eventually resting complexes). Given infinite time, all molecules that begin in transient complexes will eventually reach a resting set, and all complexes within a resting set will be visited infinitely often by any molecule that remains in the resting set. 

Our assumption of a separation of timescales allows the following
simplification to be made---transient complexes are unlikely to undergo
slower bimolecular (or higher-order) reactions, since fast reactions
rapidly transform them to other complexes. We therefore assume that only resting complexes can participate in slow reactions. This simplification avoids consideration of a large number of highly unlikely reaction pathways and intermediate complexes. For instance, it significantly reduces the problem of potentially infinite polymerization, as shown in \fig{\ref{fig:polymer}b}. Similarly, consider the example of the three-arm junction (\fig{\ref{fig:catalytic-3arm}}). Before the junction closes, displacing the initiator, it is possible for another hairpin to bind the exposed third-arm, catalyzing the formation of a fourth arm; this process could repeat, creating an implausibly large polymeric structure. We recognize that, in the limit of low concentrations, the bimolecular association will be much slower than the unimolecular branch migration. By classifying the intermediate complex as transient (to be quickly resolved by the formation of a completed 3-arm junction), we prohibit the bimolecular association reaction that could result in the polymer. 
Additionally, since the enumeration of unimolecular reactions is linear in the number of species, while
enumeration of bimolecular reactions is quadratic, eliminating the consideration of bimolecular reactions between transient complexes effectively reduces the complexity of the enumeration problem.

It should be noted that the set of reactions presented in our model still does not describe all
possible behaviors of DNA strands---notably, ``partial'' binding between
\segments is not possible, and none of the reactions
produces complexes that are \emph{pseudoknotted}. Additional details are provided in the appendix.


\section{Reaction Enumeration Algorithm} \label{sec:reaction-enumeration-algorithm}

The essence of the reaction enumeration algorithm is as follows:
\begin{enumerate}[itemsep=0mm]
\item Exhaustively enumerate complexes reachable from the initial complexes by fast reactions.
\item Classify those complexes into transient complexes or resting complexes (by calculating the SCCs in terms of fast reactions and marking those SCCs with no outgoing fast reactions resting sets).
\item Enumerate slow (bimolecular) reactions between resting complexes.
\item Repeat this process for any new complexes generated by slow (bimolecular) reactions.
\end{enumerate}

Additional details of the algorithm are provided in \sect{\ref{sec:reaction-enumeration-algorithm-details}}.


\section{Condensing reactions}\label{sec:condensing-reactions}

\newcommand{\FT}[0]{\mathbb{F}}
\newcommand{\RFT}[0]{\mathbb{R}}
\newcommand{\SC}[0]{\mathbb{S}}
\newcommand{\repr}{\sim}

It may not always be desirable or efficient to view---or even to simulate---the entire reaction
space, including all transient complexes and individual
binding/unbinding events. For this reason, we provide a mechanism of
systematically summarizing, or \term{condensing} the full network of
complexes and reactions. In this condensed representation, we represent
\term{transitions} between resting sets as a set of \term{condensed reactions,} where the reactants and products of condensed reactions are resting sets, rather than individual complexes. 
Recall that a resting set is a set of one or more complexes, connected by
fast reactions (and with no outgoing fast reactions). All of these condensed reactions will be bimolecular. 
Intuitively, if we assume that fast reactions are arbitrarily fast, then looking only at the resting sets, the condensed reactions will behave exactly the same as the full reaction network. We formalize the correspondance between trajectories in the detailed and the condensed reaction networks in \sect{\ref{sec:justification-condensed}}.

We define the \term{condensed reaction network} to be the 
reaction network $\hat{G} = (\li{Q}, \hat{\li{R}})$, where $\li{Q}$ is the set of \emph{resting sets} and $\hat{\li{R}}$ is the set of condensed
reactions. For clarity, we will refer to the ``full'', non-condensed reaction network $G = (\li{C}, \li{R})$ as the \term{\emph{detailed} reaction network}.
The goal of this section, therefore, is to describe how to calculate the condensed reaction network $\hat{G}$ from the detailed reaction network $G$. We later prove some properties of this transformation---namely that transitions are preserved between the two reaction networks in \sect{\ref{sec:justification-condensed}}.

Since much of the following discussion will involve both sets
(which may contain only one instance of each member), and multisets
(which may contain many instances of any element---for example, the
reactants and products of a chemical reaction are each multisets), we
will use blackboard-bold braces $\mset{~}$ to represent multisets and normal
braces $\set{~}$ to represent sets.  
Let us also define a useful operation on sets of multisets. Let $A$ and $B$ be sets of multisets; then we will define the \term{Cartesian sum} of $A$ and $B$ to be $A \oplus B = \set{a + b : a \in A, b \in B}$. \footnote{This operation is equivalent to taking the Cartesian product $A \times B = \set{(a, b) : a \in A, b \in B}$, then summing each of the individual pairs and giving a set containing all the sums. That is, $A \oplus B = \set{ \sum_{x \in p} x : p \in A \times B }$. The result is therefore a \emph{set of multisets}}. It is easily shown that the Cartesian sum is associative and commutative. We will therefore also sometimes write $\bigoplus_{b_i \in B} b_i$ to represent $b_1 \oplus b_2 \oplus \ldots$ for all $b_i \in B$.

We will attempt to present this section as a self-contained theory about condensed reaction networks that will be independent of the particular details of the detailed reaction enumeration algorithm. However, we will make certain restrictions on the detailed reaction network to which this algorithm is applied. The reader can verify that the reaction enumeration algorithm presented above satisfies these properties, even when the enumeration terminates prematurely (as in \sect{\ref{sec:terminating-conditions}}):
\begin{enumerate}[itemsep=0mm]
	\item{All fast reactions are unimolecular, and all slow reactions are bimolecular or higher order.}
	\item{Reactions can have any arity $(n,m)$, as long as $0 < n \le 2$ and $m > 0$.}
	\item{For any sequence of unimolecular reactions, where each reaction consumes a product of the previous reaction and the last reaction produces the original species, the sequence must consist only of 1-1 reactions.\footnote{Consider a pathological reaction network such as $a \to b + c$, $c \to a$ (such a reaction network would not be generated by our enumerator, because the number of DNA strands are conserved across reactions; this network would also \emph{not} satisfy property 3). These types of reactions prevent us from finding meaningful SCCs. }}
	\item{Reactants of non-unimolecular reactions must be resting complexes.}
\end{enumerate}

\fig{\ref{fig:condense-example}} provides an example of the reaction condensation process. Before we explain this process, let us clarify a few properties of the detailed reaction network, based on our definitions. Considering the detailed reaction network $G = (\li{C}, \li{R}_f \cup \li{R}_s)$, let $\li{R}_f^{(1,1)}$ be the fast (1,1) reactions in $\li{R}$. First note that the complexes in $\li{C}$ and the fast reactions $\li{R}_f^{(1,1)}$ form a directed graph $\Gamma = (\li{C},\li{R}_f^{(1,1)})$. We will be calculating the strongly-connected components (SCCs) of this graph using Tarjan's algorithm, as we did in the reaction enumeration algorithm. Note that \emph{every} complex (whether a resting complex or a transient complex) is a member of \emph{some} SCC of $\Gamma$. Let us denote by $\SC(x)$ the SCC of $\Gamma$ containing some complex $x$. The resting sets of $G$ are the SCCs for which there are no outgoing fast reactions of any arity. There may be (even large) SCCs that only contain transient complexes, because there is at least one outgoing fast reaction from the SCC; for instance, in \fig{\ref{fig:condense-example}}, $\set{1, 2}$ is an SCC, but is not a resting set (because of the fast reaction $2 \to 3$). Also note that many SCCs may contain only one element (for instance, $\set{3}$ is an SCC). Finally, observe that we can form another graph, $\Gamma'$, where the nodes are SCCs of $\Gamma$, and there is a directed edge between SCCs if there is a $(1,m)$ reaction with a reactant in one SCC and a product in the other. That is, $\Gamma' = (S, E)$ where $E = \set{ S(a), S(b_i) \forall b_i \in B : r = (\mset{a}, B) \in R_o \wedge \alpha(r) = (1,m) }$. $\Gamma'$ is a directed acyclic graph (since all cycles were captured by the SCCs).\footnote{It is important to note that $\Gamma'$ is not suitably described as a quotient graph on a transitive closure of fast (1,1) reactions, since $(1,m)$ reactions can also yield edges.}
 The leaves of $\Gamma'$ (those vertices with no outgoing edges) are the resting sets of $G$.

\begin{defn}
A \term{fate} $F$ of a complex $x$ is a multiset of
possible resting sets, reachable from $x$ by fast reactions. 
\end{defn}
$F$ is a multiset because, for instance, $x$ may be a dimer that can decompose
into two identical complexes, both of which are resting complexes,
such that: $x \rightarrow y + y$. Any complex $x$ may have many such
fates\footnote{Consider, for instance, the complex \cplx{17} in \fig{\ref{fig:condense-example}}; in this case, there are two possible reactions resulting in different combinations of resting sets; therefore \cplx{17} has two possible fates---either $\mset{ \set{12}, \set{13} }$ or $\mset{ \set{21}, \set{20,26} }$}, and all complexes must have at least one fate. We will denote the
set of such fates by $\FT(x)$. Computing $\FT(x)$ for all complexes $x$ in
the ensemble therefore allows each complex to be mapped to a set of
possible resting set fates. More generally, we may think about the fate 
of a multiset of complexes; this answers the question ``to what combinations
of resting sets can some molecules evolve, starting in this multiset of complexes, exclusively via 
fast reactions?'' Let us define $\FT(X)$ where $X = \mset{x_1, x_2, \ldots}$ to be:
$$\FT(X) = \bigoplus_{x \in X} \FT(x) = \FT(x_1) \oplus \FT(x_2) \oplus \ldots$$
The intuition is that, since fates for different complexes are independent, 
the set of possible fates of $X$ is the set of all possible combinations of the 
fates of $x_1$, $x_2$, etc.
From here, we can define the related notion for some reaction $r = (A, B)$, where $B = \mset{ b_1, b_2, \ldots b_n }$. 
Let $\RFT(r) = \FT(B)$, such that:
$$\RFT(r) = \FT(B) = \bigoplus_{b \in B} \FT(b) = \FT(b_1) \oplus \FT(b_2) \oplus \ldots \oplus \FT(b_n). $$ 
Finally, let $R_o(S)$ be the set of reactions leaving some SCC $S$ of $\Gamma$. 

With these definition, we can provide an expression for $\FT(x)$ in terms of a recursion:
\begin{equation}\label{eqn:fate-recursion}
\FT(x) = \begin{cases} 
\set{ \SC(x) } & \text{if } \SC(x) \text{is a resting set} \\
\bigcup_{r \in R_o\left(\SC(X)\right)} \RFT(r) & \text{otherwise} \\
\end{cases}
\end{equation}
The ``base case'' is that, if $x$ is a resting complex, then $\FT(x)$ is just $x$'s resting set. The recursive case can be evaluated in finite time, because $\Gamma'$ is a directed acyclic graph. That is, if we start with some arbitrary transient complex $x$, the recursion can be evaluated by a depth-first traversal of $\Gamma'$, starting from $x$; since $\Gamma'$ is acyclic, each branch of the depth-first traversal will terminate at a leaf of $\Gamma'$---a resting complex for which $\FT(x)$ is trivial.

Once we have computed $\FT(x)$, we can easily calculate the set of
condensed reactions. Since $\FT(x)$ captures all of the information about
the fast reactions in which $x$ participates, we must only consider slow
reactions. For each slow reaction $s = (A, B)$ in the detailed reaction network, where $A$ is the
multiset of reactant complexes and $B$ is the multiset of product complexes, 
we will generate some number of condensed reactions---one for each fate of $B$.
Specifically, the condensed reaction network $\hat{G} = ( \hat{C}, \hat{R} )$ has 
$\hat{C}$ being the set of resting sets; we build $\hat{R}$ as follows: 
for each slow reaction $s = (A,B) \in R$, 
with $\SC(A) = \mset{ \SC(a_i) : a_i \in A }$, then for each $F \in F(s)$, we add a condensed reaction $(\SC(A), F)$ to $\hat{R}$. 
In \alg{\ref{alg:rxn-condense}} we describe precisely how to calculate $\FT(x)$, and then how to compute the condensed reactions from $\FT(x)$. In \sect{\ref{sec:justification-condensed}} we present a set of theorems justifying the choice of this algorithm.

\begin{figure}
	\includegraphics{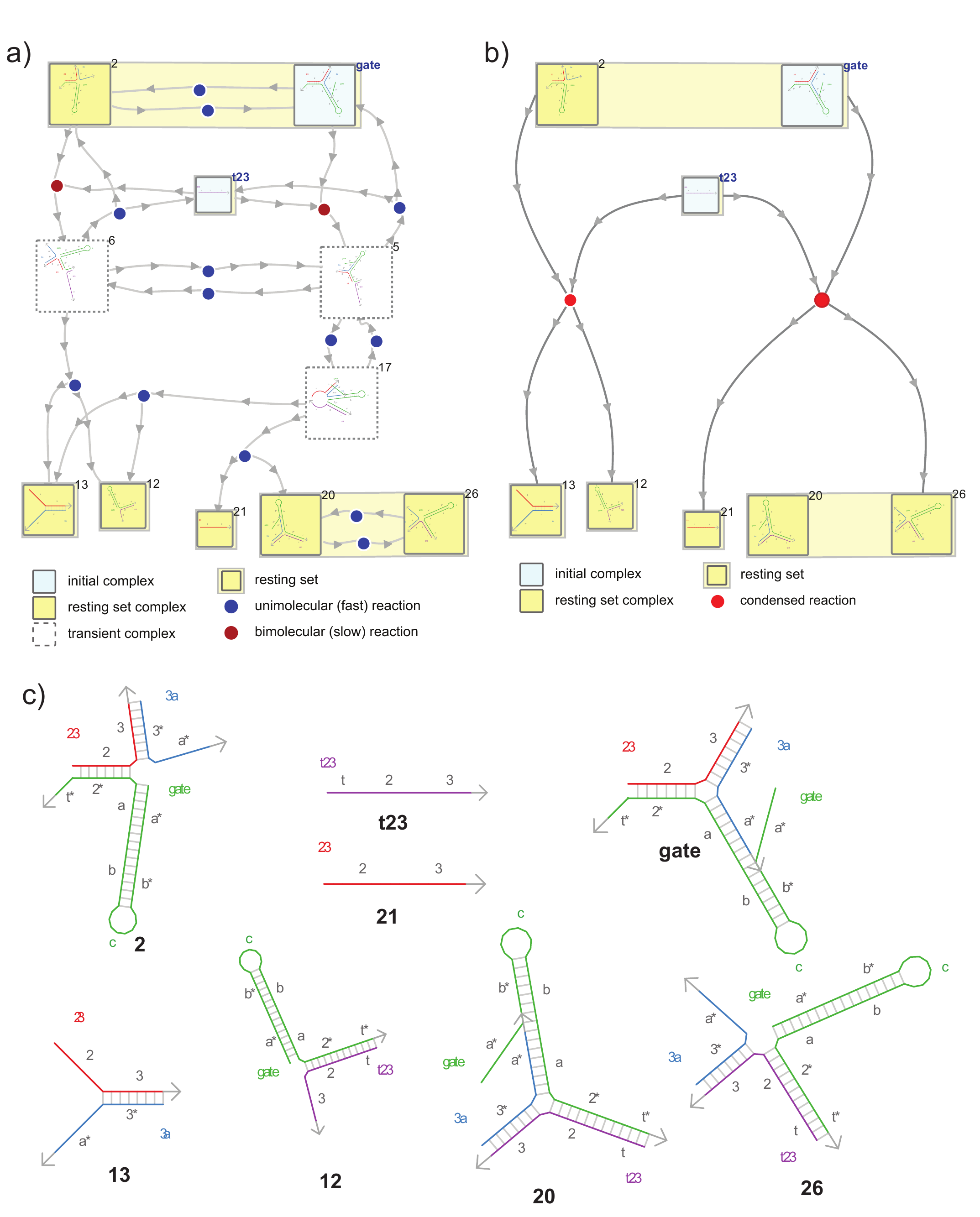} \vspace{-5mm}
	\caption[Reaction Condensation Example]{Reaction Condensation Example. In this example, we attempt to illustrate possible complexities in condensing reactions. 
	\subfigure{a} Detailed reaction network. Beginning with complex \cplx{gate}, rapid interconversion with \cplx{2} is possible. Depending on the state of the complex, interaction with \cplx{t23} yields a variety of breakdown products. 
	\subfigure{b} Condensed reaction network. The condensed reactions show two possible bimolecular reactions between \cplx{gate}/\cplx{2} and \cplx{t23}, and that each of these reactions has multiple breakdown products. Although there are two parallel pathways that yield complexes \cplx{12} and \cplx{13} in the detailed reaction network (via \cplx{6} or via \cplx{5} and \cplx{17}), only a single reaction ($\cplx{gate} \to \cplx{13} + \cplx{12}$) appears in the condensed reactions, because the two transient pathways have indistinguishable fates. 
	\subfigure{c} Key complexes shown in detail.}
	\label{fig:condense-example}
\end{figure}


\section{Approximate Kinetics}

Thus far we have only discussed the kinetics of the system in the context of our assumption that the timescales for ``fast'' and ``slow'' reactions can be separated. We have not addressed the fact that, even within the sets of fast and slow reactions, different reactions may occur at very different rates. Different rates of different reactions can greatly impact the dynamics of the system's behavior. Therefore in order to make useful predictions about the experimental kinetics of the reaction networks we are studying, we need to consider the \emph{rates} of these reactions, in addition to simply enumerating the reactions themselves. 

In this section, we present a method for approximating the rate laws governing the rates of each of the \term{move types} discussed in \sect{\ref{sec:reaction-types}}, along with an extension for calculating the rates of condensed reactions. Our implementation automatically calculates these rates and includes them in its output, allowing the user to import the resulting model into a kinetic simulator such as COPASI, provide initial concentrations for species, and simulate the resulting dynamics. It is important to emphasize that the rate laws that we provide (in particular, the formulas we give for the rate constants), although based on experimental evidence and intuition, are heuristic and approximate. Also note that the kinetics of a real physical system will be affected by parameters outside the consideration of our model---for instance, the sequences of each \segment (and therefore the binding energies) will be different, and the temperature and salt concentrations can both change the energetics of the interactions. These changes can greatly affect the real system's overall dynamic behavior. The formulas we give here are intended to roughly approximate the rate constants at \SI{25}{\degreeCelsius} and \SI{10}{\milli\Molar} Mg\textsuperscript{2+}.

Let $\rate{r}$ be the rate of some reaction $r = (A, B)$, where $A = {a_1, a_2, \ldots}$. For the sake of this section, we will assume that $\rate{r}$ is given by:
$$\rate{r} = k \prod_{a \in A} [a]$$
where $[a]$ represents the concentration of some species $a$, and $k$ a constant, which we call the \term{rate constant}. Implicit in this choice of rate law is the assumption that all reactions are elementary (meaning there is only a single transition state between the reactants and the products); in reality, this is not the case, as most DNA strand interactions have a complex transition state landscape and involve many intermediate states. Since the enumerator has no knowledge of the concentrations of any species, the problem of estimating the rate $\rate{r}$ of the reaction boils down to estimating the rate constant $k$. We will therefore attempt to provide formulas for approximating these rate constants, such that the overall reaction rates are consistent with experimental evidence in a reasonable concentration regime (sufficiently below micromolar concentrations).

In \sect{\ref{sec:detailed-rxn-kinetics}}, we provide approximate formulas for the rate constant $k$, for each of the reaction types described above. 

\subsection{Approximate condensed reaction kinetics}

\newcommand{\pdecay}[2]{\Pr{#1 \to #2}}
\newcommand{\prepr}[2]{\Pr{#1 : #2}}

Consider some reaction $\hat{r} = (\hat{A}, \hat{B})$ where $\hat{A} = \mset{\hat{A}_1, \hat{A}_2, \ldots}$ and $\hat{B}$ are multisets of resting sets. Assume $\hat{r}$ is part of a condensed reaction network $(\li{\hat{C}}, \li{\hat{R}})$, and that there is a detailed reaction network $(\li{C}, \li{R})$. Let $R_{\hat{A}}$ be the set of all detailed bimolecular reactions in $\li{R}$ between representations of $\hat{A}$:
$$R_{\hat{A}} = \set{ r = (A', B') : r \in \li{R}, A' \text{ is a representation of } \hat{A} }.$$
For bimolecular reactions (where $\hat{A} = (\hat{A}_1, \hat{A}_2)$), then $R_{\hat{A}}$ is given by:
$$R_{\hat{A}} = \set{ r = (\mset{A'_1, A'_2}, B') : r \in \li{R}, A'_1 \in \hat{A}_1, A'_2 \in \hat{A}_2 }.$$

The rate law $\rate{\hat{r}}$ for the condensed reaction $\hat{r}$ will be given by:
$$\rate{\hat{r}} = k \prod_{\hat{A}_i \in \hat{A}} [\hat{A}_i]$$
where $[\hat{A}_i]$ represents the concentration of the resting set $\hat{A}_i$, which we can assume to be in equilibrium relative to the fast reactions. We can consider $[\hat{A}_i]$ to be the sum over the concentrations of the species in $\hat{A}_i$:
$$[\hat{A}_i] = \sum_{a \in \hat{A}_i} [a].$$
Again, the real challenge will be predicting the rate constant $k$.

For simplicity, we will consider only the bimolecular case, where $\hat{A} = (\hat{A}_1, \hat{A}_2)$ (the generalization to higher order reactions is simple). The approximate rate constant $k$ is given by
\begin{equation} \label{eqn:condensed-k}
k = \sum_{r = (\mset{a_1, a_2}, B) \in R_{\hat{A}}} \Pr{a_1 : \hat{A}_1} \cdot \Pr{a_2 : \hat{A}_2} \cdot k_r \cdot \Pr{T_{B \to \hat{B}}}
\end{equation}
where: 
$\Pr{a_1 : \hat{A}_1}$ represents the probability that, at any given time, a single complex in the resting set $\hat{A}_1$ will be $a_1$, 
$\Pr{a_2 : \hat{A}_2}$ represents the probability that, at any given time, a single complex in the resting set $\hat{A}_2$ will be $a_2$, 
$k_r$ represents the rate constant for the \emph{detailed} reaction $r$, as calculated in the previous section, and 
$\Pr{T_{B \to \hat{B}}}$ represents the probability that the complexes in $B$ decay to complexes that represent $\hat{B}$. Note that, if $r$ produces products which can never be converted to the resting sets in $\hat{B}$, then this term will be 0. 

The sum is over all reactions which can represent the reactants---that is, we need to consider all possible reactions that can consume one complex from each resting set in $\hat{A}$. The overall rate for the condensed reaction will be proportional to the rates of those detailed reactions, weighted by the joint probability that those reactants are actually present, and that the products decay to the correct resting set.

To calculate $\Pr{a_1 : \hat{A}_1}$, $\Pr{a_2 : \hat{A}_2}$, and $\Pr{T_{B \to \hat{B}}}$, it is helpful to think of each of the SCCs of the the graph $\Gamma$ of the detailed reaction network as individual, irreducible Continuous Time Markov Chains (CTMCs). We will use the results for resting sets to calculate $\Pr{a_1 : \hat{A}_1}$, $\Pr{a_2 : \hat{A}_2}$, and results for SCCs of transient complexes to derive $\Pr{T_{B \to \hat{B}}}$. 

Consider first the resting sets. We can treat a single instance of the resting set $A = \set{a_1, a_2, \ldots a_L}$ to be a Markov process, periodically transitioning between each of the $L$ states, according to the reactions connecting the various complexes $a_1, a_2, \ldots$ in the resting set. We can describe the dynamics of this process by a matrix $\mat{T} \in \mathbb{R}^{L \times L}$, where the elements $T_{ij}$ of the matrix represent the rate (possibly zero) of the reaction from state $j$ to state $i$, which we denote $\rate{j \to i}$, and each diagonal element is the negative sum of the column:
\begin{equation}
T_{ij} = \begin{cases}
\rate{j \to i}                  & i \ne j \\
-\sum_{j=1, j \ne i}^{L} T_{ji} & i = j   \\ 
\end{cases}
\end{equation}
Let $\vec{s}(t) = (s_1, s_2, \ldots)^T$ be an $L$-dimensional vector giving the probabilities, at time $t$, of being in any of the $L$ states. The continuous-time dynamics of this process obeys
$$\frac{d\vec{s}}{dt} = \mat{T} \vec{s}(t).$$
For a resting set, we assume that the system has reached equilibrium, and so $\vec{s}$ is not changing with time. We therefore find the \term{stationary distribution} $\hat{\vec{s}}$ of this process by setting $d\vec{s}/dt = 0$, and recognizing that $\hat{\vec{s}}$ is the right-eigenvector of $\mat{T}$ with eigenvalue zero. Given the stationary distribution $\hat{\vec{s}} = (\hat{s}_1, \hat{s}_2, \ldots, \hat{s}_L)^T$, we recognize that 
\begin{equation}\label{eqn:prepr-stationary}
\Pr{a_i : A_i} = \hat{s}_i.
\end{equation}

Next consider the transient SCCs. To figure the probability $\Pr{T_{B \to \hat{B}}}$ that complexes in $B$ decay to complexes that represent $\hat{B}$, we will again model the SCC as a Markov process, but we cannot use the stationary distribution since the transient SCC may be far from equilibrium. Note that this SCC does \emph{not} represent a resting set, so there are additionally some $e$ outgoing fast reactions that exit the SCC. We will model this SCC, including the outgoing reactions, as an $(L + e)$-state Markov process, where each of the $e$ states is absorbing. We want to know the probability that, having entered the SCC in some state $i \in \set{1 \ldots L}$, it will leave via some reaction $j \in \set{L+1 \ldots L+e}$. We number the reactions in this way to allow them to be discussed consistently as states in the same Markov process as the complexes.

Assume the SCC is again given by $A = \set{a_1, a_2, \ldots a_L}$. Let $\mat{Q} \in \mathbb{R}^{L \times L}$ be the matrix of transition probabilities \emph{within} the SCC, such that $Q_{ij}$ is the probability that, at a given time the system's next transition is from state $i$ to state $j$, where $i,j \in \set{1 \ldots L}$. Let us therefore define an additional matrix $\mat{E} \in \mathbb{R}^{L \times e}$, where $E_{ij}$ represents the probability that the system in state $i \in \set{1 \ldots L}$ transitions next to absorbing state $j \in \set{L+1 \ldots L_e}$. We define the transition probabilities $Q_{ij}$ and $R_{ij}$ in terms of the transition rates---if $k_{ij}$ is the rate constant of the transition from state $i$ to state $j$, then 
\begin{equation}\label{eqn:transition-prob}
Q_{ij} = \frac{k_{ij}}{\sum_{j' = 1}^{L + e} k_{ij'}}
\end{equation}
and similarly for $R_{ij}$. To calculate the probability of exiting via state $j$ after entering through state $i$, we first define the \term{fundamental matrix}---which gives the expected number of visits to state $j$, starting from state $i$:
\begin{equation}\label{eqn:fundamental-matrix}
\mat{N} = \sum_{k = 0}^{\infty} \mat{Q}^k = (\mat{I}_L - \mat{Q})^{-1}
\end{equation}
where $\mat{I}_L$ is the $L \times L$ identity matrix. Let us define the \term{absorption matrix} $\mat{B} = \mat{N}\mat{R}$; the entries $B_{ij}$ represent the probability of exiting via state $j$ after entering through state $i$.

But what we really need is to calculate the relative probabilities of each of the fates of some complex $x$. Let $\pdecay{x}{F}$ be the probability that $x$ decays into a given fate $F$. 
\begin{equation} \label{eqn:pdecay-complex}
\pdecay{x}{F} = \begin{cases} 
1      & \text{if $\SC(x)$ is a resting set and $\FT(x) = \set{F}$} \\
\sum_j{B_{ij} \pdecay{r_j}{F}} & \text{if $\SC(x)$ is a transient SCC} \\
0      & \text{if $F \not\in \FT(x)$}
\end{cases}
\end{equation}
where $\mat{B} = [B_{ij}]$ is the absorption matrix for $\SC(x)$, $i$ represents the index of complex $x$ in $\SC(x)$, and $j$ is the index of the reaction exits $\SC(x)$, and $\pdecay{r_j}{F}$ is the probability that the products of the reaction $r_j$ decay to complexes that represent $F$. More formally, for some reaction $r_j = (C, D)$, where $D = \mset{d_1, d_2, \ldots}$:

\begin{equation}\label{eqn:pdecay-reaction}
\pdecay{r_j}{F} = \sum_{\set{F'_1 + F'_2 + \ldots F'_n = F}} \pdecay{d_1}{F'_1} \pdecay{d_2}{F'_2} \ldots \pdecay{d_n}{F'_n}
\end{equation}

The sum is taken over all of the ways that the fates 
$\mset{F'_k} \in \mset{\FT(d_1) \times \FT(d_2) \times \ldots : d_k \in D }$
can combine such that the $\mset{F'_k}$ sum to our target fate $F$ (that is, $\sum_{k} F'_k = F$. Within the sum, we want to calculate the joint probability that all of the complexes $d_k \in D$ decay to their respective fates $F'_k$. These terms can be computed recursively by \eqn{\ref{eqn:pdecay-complex}}.

Finally, we can write an expression for our quantity of interest. We want to know the probability $\Pr{T_{B \to \hat{B}}}$; however, we realize now that this can be computed easily using \eqn{\ref{eqn:pdecay-reaction}}:
$$ \Pr{T_{B \to \hat{B}}} = \pdecay{r}{\hat{B}},$$
where $r$ is the original, detailed bimolecular reaction.

Now we have shown how to compute all the terms necessary to evaluate \eqn{\ref{eqn:condensed-k}}. The structure of our arguments has mirrored the algorithm for deriving the condensed reactions---it is simple to efficiently evaluate \eqn{\ref{eqn:prepr-stationary}, \ref{eqn:pdecay-complex}, and \ref{eqn:pdecay-reaction}} alongside \alg{\ref{alg:rxn-condense}}, and therefore to compute a rate constant $k$ for each condensed reaction.


\section{Discussion}

We have presented an algorithm for exhaustive enumeration of hybridization reactions between a set of DNA species. Our enumerator is more flexible than previously-presented work, allowing enumeration of essentially all non-pseudoknotted structures; previous enumerators such as Visual DSD have been limited to structures without hairpins or internal loops \cite{Lakin:2012hb}. The imposed separation of timescales allows for general bimolecular reactions to be enumerated without producing a large number of physically implausible reactions and forming a (potentially infinite) polymeric structure. By instead insisting that (at sufficiently low concentrations) all unimolecular reactions will precede bimolecular reactions, we allow unimolecular reactions (for instance, ring-closing or branch migration) to terminate such polymerization reactions. The transient intermediate complex, that could otherwise participate in polymerization, is excluded from consideration of bimolecular reactions.

Further, we have demonstrated a convenient representation for condensing large reaction networks into more compact and interpretable reaction networks. We have proven that this transformation preserves the relevant properties of the detailed reaction network---namely, that all transitions between resting sets are possible in the condensed reaction network---and that the condensed reaction network does not introduce spurious transitions \emph{not} possible in the detailed reaction network. 

Our implementation exhaustively enumerates the full reaction network (or a truncated version of the network if certain trigger conditions are reached). Simulation can then be performed on the full reaction network to determine steady-state or time-course behavior. Visual DSD, includes a \term{just-in-time} enumeration mode, which combines the enumeration and simulation processes. This algorithm begins with a multiset of initial complexes, and at each step generates a set of possible reactions among those complexes; these possible reactions are selected from probabilistically to generate a new state. In this way, the model produces statistically correct samples from the continuous time Markov chain that represents the time-evolution of the ensemble. This algorithm could be an interesting future extension to our implementation.

{\bf Acknowledgements.}
The authors thank Chris Thachuk, Niles Pierce, Peng Yin, and Justin Werfel for discussion and support. This work was supported by the National Science Foundation grants CCF-1317694 and CCF-0832824 (The Molecular Programming Project).

{\footnotesize
\begin{spacing}{0.9}
\bibliography{library}
\bibliographystyle{ieeetr}
\end{spacing}
}

\clearpage
\appendix
\appendixpage

\setcounter{figure}{0}
\renewcommand{\thefigure}{S\arabic{figure}}
\setcounter{algorithm}{0}
\renewcommand{\thealgorithm}{S\arabic{algorithm}}

\section{Reaction enumeration algorithm details} \label{sec:reaction-enumeration-algorithm-details}
The reaction enumeration algorithm works by passing complexes through
a progression of several mutable sets; as new complexes are enumerated, they are added to one 
of these sets, then eventually removed and transferred to a later set. All 
complexes accumulate in either $\li{T}$ (transient complexes) or $\li{E}$ (resting state complexes). 
This progression enforces the requirement that each complex be
classified as a resting state complex or transient complex, that all
fast reactions are enumerated before slow reactions, and that complexes
are not enumerated more than once. For simplicity, we assume an operation
$\Call{Pop}{S}$ exists that removes and returns some element from the mutable set $S$.

\begin{itemize}
\item
	$\li{B}$ contains complexes that have had no reactions enumerated yet.
	Complexes are moved out of $\li{B}$ into $\li{F}$ when their
	\term{neighborhood} is considered. We define a \term{neighborhood} about 
	a complex $c$ to be the set of complexes that can be produced by a series of zero or more 
	fast reactions starting from $c$.

\item
	$\li{F}$ contains complexes in the current neighborhood which have
	not yet had fast reactions enumerated. These complexes will be moved
	to $\li{N}$ once their fast reactions have been enumerated.
\item
	$\li{N}$ contains complexes enumerated within the current neighborhood, 
	but that have not yet been characterized as transient or resting states. 
	Each of these complexes is classified, then moved into $\li{S}$ or $\li{T}$.
\item
	$\li{S}$ contains resting state complexes which have not yet had
	bimolecular reactions with set $\li{E}$ enumerated yet. All
	self-interactions for these complexes have been enumerated.
\item
	$\li{E}$ contains enumerated \emph{resting state complexes}. Only
	cross-reactions with other end states need to be considered for these
	complexes. These complexes will remain in this list throughout
	function execution.
\item
	$\li{T}$ contains \emph{transient complexes} which have had their fast
	reactions enumerated. These complexes will remain in this list
	throughout function execution.
\end{itemize}

\fig{\ref{fig:enum-lists}} summarizes the progression of complexes through these sets. Additionally, two other sets are accumulated over the course of the enumeration:

\begin{itemize}
\item
	$\li{R}$ contains all reactions that have been enumerated.
\item
	$\li{Q}$ contains all resting states that have been enumerated.
\end{itemize}

\begin{figure}[h]
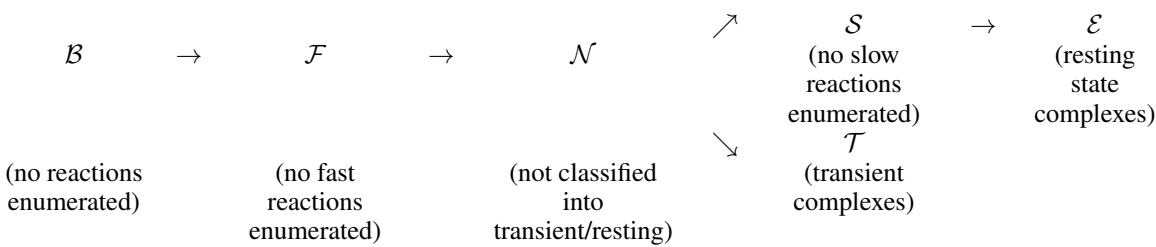

		\begin{tabulary}{\textwidth}{ClClClClC}
		~                         & ~             & ~                              & ~             & ~                                       & $\nearrow$ & $\li{S}$                       & $\rightarrow$ & $\li{E}$                  \\
		$\li{B}$                  & $\rightarrow$ & $\li{F}$                       & $\rightarrow$ & $\li{N}$                                & ~          & (no slow reactions enumerated) & ~             & (resting state complexes) \\
		~                         & ~             & ~                              & ~             & ~                                       & $\searrow$ & $\li{T}$                       & ~             & ~                         \\
		(no reactions enumerated) & ~             & (no fast reactions enumerated) & ~             & (not classified into transient/resting) & ~          & (transient complexes)          & ~             & ~                         \\                            & ~             & ~                         \\
		\end{tabulary}
\caption[Passage of complexes through enumerator]{Passage of complexes through various lists; complexes begin with no outgoing reactions enumerated ($\li{B}$), then have fast reactions in their neighborhood enumerated ($\li{F}$, $\li{N}$) before being classified into transient ($\li{T}$) or resting state ($\li{S}$) complexes; finally slow reactions are enumerated ($\li{E}$).}
\label{fig:enum-lists}
\end{figure}

In \alg{\ref{alg:rxn-enum}} we describe the reaction enumeration algorithm in detail.

\subsection{Terminating conditions}\label{sec:terminating-conditions}

Although our enumerator is designed to avoid enumerating implausible polymerization reactions, as described in \fig{\ref{fig:polymer}}, it is possible to enumerate systems which result in genuine polymerization, such as those described by \cite{Dirks:2004cb,Venkataraman:2007bw}. To allow such enumerations to terminate, our enumerator places a soft limit on the maximum number of complexes and the maximum number of reactions that can be enumerated before the enumerator will exit. These limits are checked before the neighborhood of fast reactions is enumerated, for each complex in $B$. The limits are configurable by the user.

If the number of complexes in $\li{E} \cup \li{S} \cup \li{T}$ is greater than the maximum number of complexes, or the number of reactions in $\li{R}$ is greater than the configured maximum, the partially-enumerated network is ``cleaned up'' by deleting all reactions in $\li{R}$ that produce complex(es) leftover in $\li{B}$. That is, no complex will be reported in the output that has not had its neighborhood of fast reactions enumerated exhaustively. Evaluating the limits only \emph{between} consideration of each neighborhood (rather than during) prevents the pathological mis-classification of resting state and transient complexes, but means the limit may be exceeded if the last neighborhood considered is large.

\section{Justification of the condensed reaction algorithm}\label{sec:justification-condensed}

We will now justify the algorithm for condensing reactions with several theorems that show the
relationship between the condensed reaction network $\hat{G} = (\li{\hat{C}},
\li{\hat{R}})$ and the detailed reaction network $G = \li{C}, \li{R}$. To do this,
we introduce two definitions.  

First, we need a notion of what kind of processes
from the detailed  reaction network are actually included in the condensed
reaction network. 
We define a \term{fast transition} $T_{\mset{x} \to B}$ to be a sequence of (zero or more) unimolecular reactions that begin from a single initial (transient or resting state) complex $x$ and result in a multiset $B$ of resting state complexes.
A \term{resting state transition} $T_{\mset{a_1, a_2} \to B}$ is a sequence of detailed reactions starting with a bimolecular (slow) reaction (by definition between two resting state complexes $a_1$ and $a_2$), followed by a sequence of (zero or more) unimolecular reactions that can occur if the system starts with just $a_1$ and $a_2$ 
present, and such that the final state $B$ consists exclusively of resting state complexes.

Second, we need a notion of correspondence between some reaction in the condensed 
reaction network and a \emph{transition} that can occur in the detailed reaction network.
For some multiset of resting states
$\hat{A} = \mset{\hat{A}_1, \hat{A}_2, \ldots}$, where each
$\hat{A_i} = \set{a_{i,1}, a_{i,2}, \ldots}$, a \term{representation} of
$\hat{A}$ is a set containing a choice of \emph{one} complex $a_{i,j}$ from each
$\hat{A}_i$. Note that if any of the sets $\hat{A_i} \in \hat{A}$ are not singletons,
then there are multiple representations of $\hat{A}$. 
For example, if $\hat{A} = \mset{ \hat{A}_1, \hat{A}_2 }$, 
$\hat{A}_1 = \set{ a_{1,1}, a_{1,2} }$, and $\hat{A}_2 = \set{ a_{2,1}, a_{2,2} }$, then there
are four possible representations of $\hat{A}$: 
$\set{ a_{1,1}, a_{2,1} }$, $\set{ a_{1,2}, a_{2,1} }$,
$\set{ a_{1,1}, a_{2,2} }$, or $\set{ a_{1,2}, a_{2,2} }$.
We can write $A \repr \hat{A}$ to indicate that $A$ is a representation of $\hat{A}$\footnote{Note that $\repr$ itself is not an equivalence relation, since the left-hand side (multisets of complexes) and the right-hand side (multisets of resting states) are not members of the same set and therefore neither symmetry not reflexivity hold. It could be said that each of the resting states form an equivalence class, and the set $\li{Q}$ of resting states is the quotient space of this equivalence class. However, the DAG $\Gamma'$ is not simply the quotient graph of $\Gamma$ (the graph between complexes, connected by (1,1) reactions) under this equivalence relation, because $(1,2)$ reactions are not represented in $\Gamma$, yet must still generate possible fates in $\Gamma'$.}

\begin{lemma}\label{lem:fate-to-transition} 
For every complex $x$, and for every fate $F$ in the set of fates $\FT(x)$, and for every $B$ such that $B \repr F$, there exists a fast transition $T_{ \mset{x} \rightarrow B }$.
\end{lemma}
\begin{proof}
Consider a single fate $F \in \FT(x)$. In the base case where $x$ is a resting complex, then $\FT(x) = \set{ \SC(x) }$ is singleton, and we take $F = \SC(x)$. If $\SC(x)$ is non-singleton, then any transition between $x$ and another complex in $b \in \SC(x)$ will satisfy the property that $B = \mset{b} \repr F$. If $\SC(x)$ is singleton ($\SC(x) = \set{x}$), then the transition is degenerate $T_{\mset{x} \to \mset{x}}$, but still satisfies the propery that $B = \mset{x} \repr F$.

When $x$ is not a resting state complex, recognize that each fate $F \in \FT(x)$ was generated by application of the recursive case of \eqn{\ref{eqn:fate-recursion}}, in which a union is taken over outgoing reactions from $\SC(x)$. That is, each fate $F \in \FT(x)$ is generated by some outgoing reaction $r = (\alpha, \beta)$ from $\SC(x)$. Specifically, $F$ is one element of the set $\RFT(r) = \FT(\beta) = \bigoplus_{b \in \beta} \FT(b)$. For any $B \repr F$, the fast transition $T_{\mset{x} \to B}$ can thus be accomplished by first following $r$, 
followed by the concatenation of $T_{ \mset{b} \to \FT(b) }$ for each $b \in \beta$.

By induction, we recognize that, for any complex $x$ and fate $F \in \FT(x)$, a detailed fast transition can be accomplished from $x$ to $B \repr F$.
\end{proof}

\begin{theorem} (Condensed reactions map to detailed reactions) 
For every condensed reaction $\hat{r} = ( \hat{A}, \hat{B} )$, for every $A$ that represents $\hat{A}$, and for every $B$ that represents $\hat{B}$, there exists a detailed resting state transition $T_{A \to B}$.
\end{theorem}
\begin{proof} 
First, recognize that every condensed reaction $\hat{r} = (\hat{A}, \hat{B})$ was generated by some bimolecular reaction $r = (A, A')$, where $A$ contains only resting state complexes and represents $\hat{A}$. Therefore we must only show that there exists a fast transition $T_{A' \to B}$, such that $B \repr \hat{B}$. We recognize that the multiset of products $\hat{B}$, of the condensed reaction $\hat{r}$, was generated from one element of $\RFT(r) = \FT(A')$. Therefore, $\hat{B}$ is an element of $\FT(A')$. By \lem{\ref{lem:fate-to-transition}}, there exists a detailed transition $T_{A' \to B}$, such that $B \repr \hat{B}$. Therefore there exists a transition $T_{A \to B}$ such that $A \repr \hat{A}$ and $B \repr \hat{B}$.
\end{proof}

\begin{lemma}\label{lem:transition-to-fate} For some complex $x$, each fast transition $T_{ \mset{x} \to B }$, such that $B$ contains only resting state complexes, corresponds to exactly one fate $F \in \FT(x)$. Specifically, there exists some fate $F \in \FT(x)$ such that $B \repr F$.\end{lemma}
\begin{proof}
Consider the base case where $x$ is a resting state complex; in this case, all fast transitions from $x$ must lead to another resting state complex in $\SC(x)$. $\FT(x) = \set{ \SC(x) }$, by \eqn{\ref{eqn:fate-recursion}}, and therefore this transition corresponds to the fate $F = \SC(x)$. 

Consider some detailed fast transition $T_{ \mset{x} \to B }$ such that $B = \mset{b_1, b_2, \ldots}$ contains only resting state complexes. We recognize that, if $x$ is \emph{not} a resting state complex, there must be at least one reaction in this process. The transition begins with this initial reaction $r^0 = (\mset{x}, Y)$; $Y$ may have multiple products, each of which decays independently to some complex or set of complexes in $B$.

Realize that, for some reaction $r_i = (A_{i-1},A_{i})$, by applying \eqn{\ref{eqn:fate-recursion}} we recognize that if a fate $F$ is reachable from $A_{i}$, then it is reachable from $A_{i-1}$. That is, for some fate $F$, $F \in \FT(A_{i}) \implies F \in \FT(A_{i-1})$. This means that, for some prior reaction $r_{i-1} = (A_{i-2}, A'_{i})$ such that $A_{i} \subseteq A'_{i}$---that is, a reaction $r_{i-1}$ that produces the reactant of $r_{i}$---$F \in \RFT(r_i) \implies F \in \RFT(r_{i-1})$. 

Next, we note that the set of products $B$ of the transition $T_{ \mset{x} \to B }$ must represent some fate; that is, $B \repr F$. Since $B$ consists exclusively of resting states, $F = \mset{ \SC(b) : b \in B }$. Multiple reactions $r_1, r_2, \ldots r_m$ may have produced the complexes in $B$, let us denote this set $R_B$: $R_B = \set{r_i = (A_i, B_i) \in T_{\mset{x} \to B} : B_i \subseteq B }$; $B$ is therefore the sum of the products of these reactions: $B = \sum_{r_i = (A_i, B_i) \in R_B} B_i$. Because $B \repr F$ and \eqn{\ref{eqn:fate-recursion}} includes all possible sums of $\RFT(B)$, this means that if we choose fates $F_i \in \FT(r_i)$ for each of those reactions, there exists some set $\set{F_1, F_2, \ldots, F_m}$ such that $F_1 + F_2 + \ldots + F_m = F$. 

Consider one of the reactions $r_i = (A_i, B_i) \in R_B$, that produces complex(es) in $B$. Each fate $F' \in \RFT(r_i)$ is \emph{also} a fate of any reaction that produces $A_i$. This means that, for $r_i$, the particular fate $F_i \in \set{F_1, F_2, \ldots, F_m}$ satisfying $\sum_{j = 1}^{m} F_j$ must \emph{also} be a fate of any reaction that produces $A_i$. By induction, we can work backwards from $r_i$ all the way to the initial reaction $r^0$, and recognize that $F_i \subseteq F^0$ for some $F^0 \in \RFT(r^0)$. The same is true for all reactions $r_i \in R_B$. Because the recursive case of \eqn{\ref{eqn:fate-recursion}} sums over all combinations of fates for all such pathways, the $F_1 + F_2 + \ldots + F_m = F$ must be a member of $\RFT(r^0)$, and therefore a member of $\FT(x)$.
\end{proof}

\begin{theorem} (Detailed reactions map to condensed reactions) For every detailed resting state transition $T_{A \to B}$, there exists a condensed reaction $\hat{r} = (\hat{A}, \hat{B})$ such that $A$ represents $\hat{A}$ and $B$ represents $\hat{B}$. \end{theorem}
\begin{proof} 
Since $T_{A \to B}$ is a transition between two sets ($A$ and $B$) of detailed resting state complexes, the transition consists of two steps: first, a bimolecular reaction $r = (A, A')$ converts $A$ to $A'$; second, a series of unimolecular reactions convert the complexes in $A'$ to $B$. The algorithm generates one or more condensed reactions for each detailed bimolecular reaction. Specifically, the algorithm generates one condensed reaction for each combination of fates of the products in $A'$. That is, each of the condensed reactions is generated from one element in $\RFT(r) = \bigoplus_{a' \in A'} \FT(A')$. By \lem{\ref{lem:transition-to-fate}}, for each product $a' \in A'$, $\FT(a')$ corresponds to the set of possible transitions from $a'$ that result in some resting state. Therefore we can choose any possible fast transition between $T_{A' \to B}$, and it will correspond to some element of $\RFT(r)$---and therefore to a condensed reaction $\hat{r} = (\hat{A}, \hat{B})$.
\end{proof}

Intuitively, these two theorems mean that the condensed reaction network
effectively models the detailed reaction network, at least in terms of
transitions between resting states. The first theorem shows that a
condensed reaction must be mapped to a suitable sequence of reactions in the
detailed reaction network. 
The second theorem shows the converse---that
any process in the detailed reaction network is represented by the
condensed reactions.
Having proved these theorems, we propose the following corollaries that extend this reasoning from individual (detailed and condensed) reactions to sequences of condensed and detailed reactions. We omit the proofs.

\begin{corollary}
For any sequence of condensed reactions starting in some initial state $\hat{A}$ and ending in some final state $\hat{B}$, and for any $A \repr \hat{A}$ and for any $B \repr \hat{B}$, there exists a sequence of detailed reactions starting in $A$ and ending in $B$.
\end{corollary}
Proof omitted.

\begin{corollary}
Conversely, for any sequence of detailed reactions starting in some multiset of resting state complexes $A$ and ending in some multiset of resting state complexes $B$, there exists a sequence of condensed reactions starting in $\hat{A}$ and ending in $\hat{B}$ such that $A \repr \hat{A}$ and $B \repr \hat{B}$.
\end{corollary}
Proof omitted.


\section{Approximate detailed reaction kinetics} \label{sec:detailed-rxn-kinetics}

\begin{description}

	\item[(2,1) Binding] --- $k = \SI{1e6}{\per\Molar\per\second}$, according to experimental evidence\cite{Zhang:2009ce,Morrison:1993uk}.
	
	\item[(1,1) Binding] --- The energetics of 1,1 binding depend strongly on whether the reaction forms an entropically unfavorable ``internal loop'' or ``bulge.'' We therefore provide three separate rate constant formulas for these different conditions. Zipping is binding between two adjacent \segments (e.g. between $d_{i,j}$ and $d_{i,j+1}$). Hairpin closing is where two non-adjacent \segments bind to form a hairpin. Bulge closing is where two non-adjacent \segments bind to form a more complex internal loop.
		\begin{description}
			\item[Zipping] --- $k = 10^{8} / \ell~\si{\per\second}$, where $\ell$ is the length of the \segment. Estimates range from \SI{e-7}{\second} to \SI{e-8}{\second} for the zipping time of a single base pair.\cite{Wetmur:1968vj,Porschke:1974vp}
			
			\item[Hairpin closing] --- $k = (\num{2.54e8}) (\ell + 5)^{-3} ~\si{\per\second}$, where $\ell$ is the length of the hairpin. Bonnet et al. measured various hairpin closing rates and determine empirically that the length dependence is approximately given by 
			$$k = a (\ell + 5)^{-c}$$
			where the exponent $c$ depends on the temperature and the parameter $a$ is to be fitted to the data\cite{Bonnet:1998tt}. For \SI{25.7}{\degreeCelsius}, the authors report $c = 3$ provides the best fit. We estimated $a$ by a linear fit of $(\ell+5)^{-3}$ to experimentally measured values of $k$, using the following data from Fig. 7 of \cite{Bonnet:1998tt}:

			\begin{table}[h]
		    \centering
			\begin{tabular}{lll}
			$\ell$ & $(\ell+5)^{-3}$ & $k$  \\
			\hline
			30 & $(30+5)^{-3}$ & 5000  \\
			21 & $(21+5)^{-3}$ & 10000 \\
			16 & $(16+5)^{-3}$ & 20000 \\
			12 & $(12+5)^{-3}$ & 50000
			\end{tabular}
			\end{table}

			\item[Bulge closing] --- $k = (\num{2.54e8}) (\ell' + 5)^{-3} ~\si{\per\second}$, where $\ell' = |y| + |w| + 5$ approximates the size of the bulge. We used the same expression as for hairpins, but must account for the additional length from the region of unpaired bases in the bulge. Consider binding between two strands, as in \fig{\ref{fig:bulge-formation}}. Let $\ell' = |y| + |w| + 5$, such that the rate constant is given by $k = a (\ell' + 5)^{-3}$.

			\begin{figure}[h]
			\centering
			\includegraphics{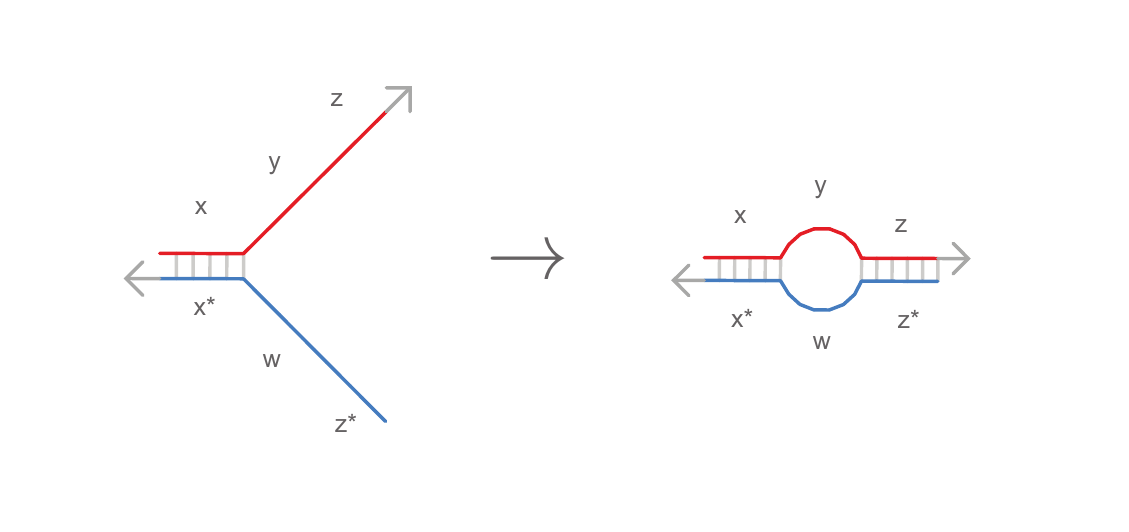}
			\caption{Bulge formation}
			\label{fig:bulge-formation}
			\end{figure}

			\item[Multiloop closing] --- $k = (\num{2.54e8}) (\ell' + 5)^{-3} ~\si{\per\second}$, where $\ell' = \sum_{i=1}^{m} |d_i| + 5(m-1)$ approximates the size of the multiloop containing \segments $d_1, d_2, \ldots d_m$. We used a similar expression as for bulges, but must generalize to account for the presence of $m$ stems and loop segments within the multiloop. We therefore treat the multiloop size as the length of each \segment, plus 5 nucleotides for each stem.
			
		\end{description} 

	\item[Opening] --- $k = 10^{6 - 1.24 \ell}$. Note that 2,1 binding is the reverse reaction of 1,2 opening. Experimental evidence suggests this ratio depends exponentially on the length $\ell$ of the binding/opening \segment\cite{Wetmur:1991dh}, and therefore:
	$$e^{\Delta G^\circ/RT} = \frac{k_\text{reverse}}{k_\text{forward}} = \frac{\text{1,2 opening rate}}{\text{2,1 binding rate}} = 10^{-a \ell}~\si{\Molar}$$
	where $\Delta G^\circ$ is the Gibbs free energy at standard conditions, $T$ is the temperature in Kelvin, $R$ is the ideal gas constant, and $a$ is an unknown constant. Given the binding rate constant $k_\text{forward} = \SI{1e6}{\per\Molar\per\second}$, this means that the opening rate constant $k_\text{reverse} = 10^{(6 - a \ell)}~\si{\per\second}$. We know that the typical energy of a single base stack is $-1.7$ \si{\kcal\per\mole} at $T = 298$ \si{\K} (25 \si{\degreeCelsius}) in a typical salt buffer\cite{Wetmur:1991dh}. We can use this information to solve for the constant $a$:
	$$e^{\Delta G^\circ/RT} = a 10^{-\ell} \implies \Delta G^\circ = (-RT \ln 10) a \ell$$
	$$\Delta G^\circ = - 1.7 \ell \implies -1.7 \ell = (-RT \ln 10) a \ell$$
	$$\implies a = 1.7 / (RT \ln 10) \approx 1.24$$
	for $R$ = \SI{1.99e-3}{\kcal\per\K\per\mol}, $T$ = \SI{298}{\K}.

	\item[3-way branch migration] --- We can consider the 3-way branch migration process as consisting of an initiation step that takes some time $a~\si{\second}$ and a series of branch migration steps; the time for branch migration steps varies quadratically with the length $\ell$ of the \segment\cite{Zhang:2009ce,Srinivas:2013jo}, with some timescale $b~\si{\second}$ each. The expected time $\Ex{t}$ for 3-way branch migration of $\ell$ bases, including an initiation time $b$, is 
	$$\Ex{t} = a + b \ell^2.$$
	Imperfectly approximating the completion time of branch migration as a Poisson rate with the same expected time, the rate constant $k$ is therefore
	$$k = \frac{1}{\Ex{t}} = \frac{1}{a + b \ell^2}$$
	where the values of $a$ and $b$ must be determined experimentally, and represent the expected time for initiation and individual branch migration steps, respectively\cite{Srinivas:2013jo}. We assume that the values are different for branch migration between adjacent \segments ($d_{i,j}$ and $d_{i,j+1}$) and for remote toehold-mediated branch migration\cite{Genot:2011gt}, where the \segments are non-adjacent. 
		\begin{description}
			\item[Adjacent] --- $k = \left[(\num{2.8e-3}) + (\num{0.1e-3}) \ell^2 \right]^{-1}~\si{\per\second}$. Srinivas et al. measured $a = \SI{2.8e-3}{\second}$ and $b = \SI{0.1e-3}{\second}$\cite{Srinivas:2013jo}.
			\item[Remote toehold-mediated] --- $k = \left[ \rho(\ell) (\num{2.8e-3}) + (\num{0.1e-3}) \ell^2 \right]^{-1}~\si{\per\second}$ We assume the same step rate $b$ as for adjacent branch migration above, but assume that the initiation time $a$ is a factor of $\rho(\ell')$ slower for remote toehold-mediated branch migration than for adjacent branch migration. We define $\rho(\ell) = \alpha/k_\text{multiloop closing}(\ell')$, where $k_\text{multiloop closing}(\ell')$ is calculated by the rate constant formula given above for multiloop closing over a bulge length $\ell'$, and $\alpha$ is a constant for remote toeholds, fitted to the experimental data.
		\end{description}

	\item[4-way branch migration] --- $k = \frac{1}{77 + \ell^2}~\si{\per\second}$. We use a similar expression to that for 3-way branch migration. Dabby measured $a = \SI{77}{\second}$ and $b = \SI{1}{\second}$ \cite{Dabby:2013uq}. 

\end{description}


\section{Nucleic acid secondary structure and pseudoknots}

Nucleic acid structures can be divided into two classes: those with base pairs that can be nested into a tree-like structure are called ``non-pseudoknotted'' (\fig{\ref{fig:structures}}), while those that do not obey this nesting property are ``pseudoknotted'' \cite{Staple:2005ga,Liu:2010he} (\fig{\ref{fig:pseudoknots}}). There are vastly more possible pseudoknotted structures than non-pseudoknotted structures, so allowing enumeration of pseudoknotted intermediates greatly increases the possible size of the generated reaction network. It would be attractive to enumerate only a subset of pseudoknotted secondary structures, but a poor understanding of the energetics of pseudoknotted structures makes it very hard to determine \emph{which} pseudoknots to allow and which to omit\cite{Turner:1988gl,Mathews:2004wz,isambert2000modeling}. Further, reactions (such as three-way branch migration) that are always plausible for non-pseudoknotted structures can yield topologically-impossible structures for pseudoknotted intermediates. For the sake of simplicity and efficiency we will therefore restrict our attention to non-pseudoknotted intermediates, for the time being. 

\begin{figure}
\includegraphics{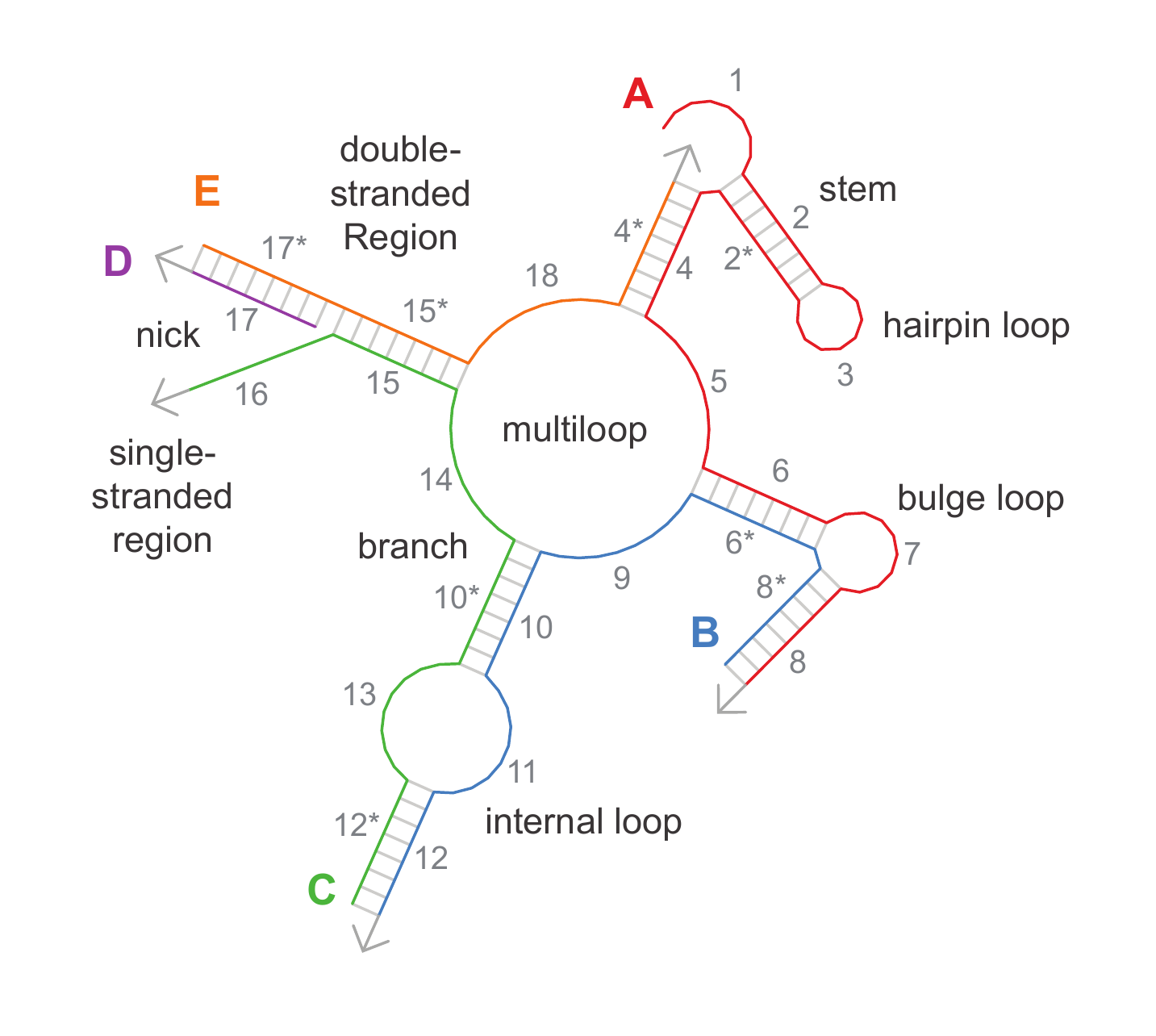}
\caption[Non-pseudoknoted structures]{Non-pseudoknoted secondary structures. Our enumerator is capable of handling a wide range of non-pseudoknotted secondary structures, depicted above.}
\label{fig:structures}
\end{figure}

\begin{figure}
\includegraphics{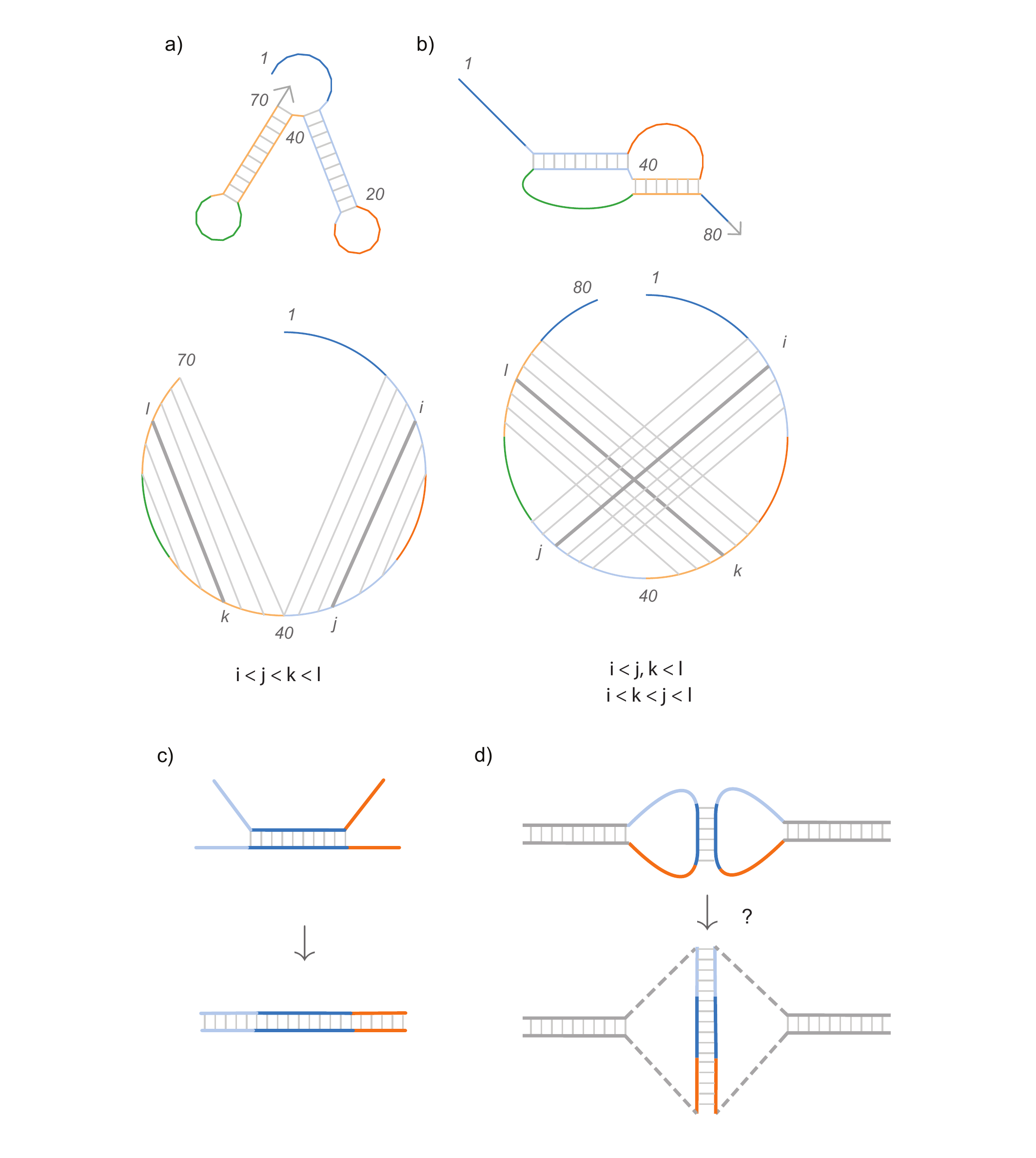}
\caption[Pseudoknotted structures]{Pseudoknotted vs. non-pseudoknotted structures. \subfigure{a} Structure containing no pseudoknots; all base pairs are either nested ($i < k < l < j$), or adjacent ($i < j < k < l$). \subfigure{b} Structure containing a pseudoknot; $i < j$ and $k < l$, but $i < k < j < l$, so the nesting property is not satisfied. This can be seen geometrically in the circle diagram. \subfigure{c} For non-pseudoknotted complexes, 1-1 binding is always permissible and yields plausible structures. \subfigure{d} For a pseudoknotted complex such as this \term{kissing-loop} motif, binding between adjacent domains may produce unrealistic structures (e.g. binding between adjacent domains in the kissing-loop duplex would likely require partial unwinding of the duplexes). Lack of a coherent and efficient energetic model for structures with pseudoknots makes the plausibility of these structures difficult to evaluate.}
\label{fig:pseudoknots}
\end{figure}


\section{Limitations} \label{sec:limitations}

As mentioned in the main text, there are two major limitations to the enumerator, as presented here.

The first limitation precludes several behaviors, such as partial binding between
similar \segments with small numbers of nucleotide mismatches (a concept which has 
been exploited to precisely control the thermodynamics and therefore 
specificity of DNA hybridization reactions \cite{Zhang:2012ct}), as well as 
partial binding at the edges between two adjacent \segments. 

The second limitation precludes consideration of the wide range of structures
that contain pseudoknots.  The distinction between pseudoknotted and non-pseudoknotted complexes is demonstrated in \fig{\ref{fig:pseudoknots}}
Essentially, a pseudoknotted complex is one which contains two intercalated stem-loop
structures. Consider a complex with $N$ nucleotides, each numbered $1, 2, \ldots, N$. 
Formally, this is \emph{non-pseudoknotted} if, for every
two base pairs $i, j$ and $k, l$ (where $i, j, k, l$ are the numerical
indices of four nucleotides in a complex) where $i < j$, $k < l$, and
$i < k$, then either $i < k < l < j$ (that is, the pair $k, l$ is nested
within the pair $i, j$) or $i < j < k < l$ (that is, the pairs $i, j$
and $k, l$ occur sequentially). Pseudoknotted structures do not obey
this nesting property. This can be easily seen by drawing the DNA backbone
of a pseudoknotted structure around a circle, and drawing base pairs as chords 
of the circle; if these chords cross, a complex is pseudoknotted..

It should be apparent that there are
combinatorially more pseudoknotted structures than non-pseudoknotted
structures. The presence of pseudoknots has important implications for
the energetics of a structure---notably, many parameters of the
energetic models that incorporate pseudoknots have not been measured and
as a consequence it is difficult to accurately calculate the free energy
and therefore estimate the stability of pseudoknotted structures. Additionally,
certain reactions have non-intuitive behavior; for instance, consider a 1-1 
binding reaction between adjacent, complementary domains. For non-pseudoknotted
structures, this reaction is always permissible. For pseudoknotted complexes,
it is easy to consider scenarios where a naive branch migration would lead to
physically implausible structures (see \fig{\ref{fig:pseudoknots}}).
For simplicity and efficiency, we therefore do not consider pseudoknotted complexes.



\section{Implementation}

We have implemented the reaction enumeration algorithm, as well as the algorithm for condensing reactions, in a command-line utility written in Python. The enumerator allows input to be specified using the Pepper Intermediate Language (PIL)---a flexible, text-based language for describing DNA strand displacement systems \cite{Ligocki:2010wc}---as well as a simple, JSON-based interchange format. The enumerator can produce both full and condensed reaction spaces for a set of starting complexes. It can also write generated systems to several output formats: 

\begin{description}
\item [Pepper Intermediate Language (PIL)] --- Output format resembles input format \cite{Ligocki:2010wc}, but additional enumerated complexes and reactions are added.
\item [Systems Biology Markup Language (SBML)\cite{Hucka:2003wg}] --- Standard format used by many modeling packages in systems and synthetic biology. Allows models to be transferred to numerous simulation and analysis utilities. 
\item [Chemical Reaction Network (CRN)] --- Simple plain text format listing each reaction on a single line. In this format, some reaction $A + B \to C + D$ would be written \texttt{A + B -> C + D}. 
\item [ENJS] --- JSON-based output format; ENJS can be visualized interactively using \name---this visualization is discussed in Chapter 2, and examples are shown in Chapter 6 of ref. \cite{Grun:2014vd}.
\end{description}


\section{Supplemental Pseudocode}

	\begin{algorithm}[H]
	\caption{Bind 1-1 move type}
	\label{alg:bind11}
	\begin{algorithmic}[1]
		\Function{bind11}{$c : \type{Complex}$}
			\State $R \gets \set{~}$ \Comment{Initialize empty set of reactions}
			\State $c = (S, T), S = \set{s_1, s_2, \ldots s_{|S|}}$
			\ForAll{$i \in \set{1 \ldots |S|}$} \Comment{For each strand}
				\State $s_i = (z, D)$
				\ForAll{$j \in \set{1 \ldots |D|}$} \Comment{For each \segment}
					\If{$T_{i,j} = \varnothing$} \Comment{If \segment $d_j$ is unpaired}
						\State $i' \gets i$
						\State $j' \gets j+1$
						\While{$i' < |S|$ and $(i',j') \ne (i,j)$} \Comment{Iterate over all \segments with higher indices}
							\State $s_{i'} = (z', D'), D' = \set{d'_1, d'_2, \ldots}$
							\If{$j' \ge |D'|$} \Comment{If at the end of a strand}
								\State $(i',j') \gets (i' + 1, 0)$ \Comment{Move to first \segment on next strand}
							\ElsIf{$T_{i', j'} \ne \varnothing$} \Comment{If \segment is paired}
								\State $(i', j') \gets T_{i',j'}$ 
								\Statex \Comment{Skip over the internal loop between $i,j$ and $i',j'$ so we don't make a pseudoknot}
							\ElsIf{$d_j = {d'}_{j'}^*$} \Comment{If $d_j$ is complementary to $d_j'$}
								\State $T' \gets T$ \Comment{Make a new structure}
								\State $T'_{i,j} \gets (i', j')$, $T'_{i',j'} \gets (i,j)$ \Comment{Where \segments $(i,j)$ and $(i',j')$ are paired}
								\State $c' \gets (S, T')$
								\State $R \gets R \cup \set{ (\mset{c}, \mset{c'}) }$ \Comment{Make a new reaction that yields this complex}
							\EndIf
							\State $j' \gets j' + 1$ \Comment{Go to next \segment}
						\EndWhile
					\EndIf
				\EndFor
			\EndFor
			\State \Return $R$
		\EndFunction
	\end{algorithmic}
	\end{algorithm}

	\begin{algorithm}[H]
	\caption{Bind 2-1 move type}
	\label{alg:bind21}
	\begin{algorithmic}[1]
		\Function{bind21}{$c : \type{Complex}$, $c' : \type{Complex}$}
			\State $R \gets \set{~}$ \Comment{Initialize empty set of reactions}
			\State $c = (S, T)$, $c' = (S', T')$
			\ForAll{$d$ such that $\exists (z,D) \in S$, $d \in D$} \Comment{For each \segment on each strand in $c$}
				\ForAll{$d'$ such that $\exists (z',D') \in S'$, $d' \in D'$} \Comment{For each \segment on each strand in $c'$}
					\If{$d' = d^*$} \Comment{If $d$ and $d'$ are complementary}
						\State $S'' \gets S \cup S'$
						\State $T'' \gets $ combine structures, pair $d$ and $d'$
						\State $c'' \gets (S'', T'')$ \Comment{Make new complex that joins $c$ and $c'$}
						\State $R \gets R \cup \set{ ( \mset{c, c'}, \mset{c''} ) }$ \Comment{Make reaction that yields this complex}
					\EndIf
				\EndFor
			\EndFor
			\State \Return $R$
		\EndFunction
	\end{algorithmic}
	\end{algorithm}

	\begin{algorithm}[H]
	\caption{Opening move type}
	\label{alg:open}
	\begin{algorithmic}[1]
		\Function{open}{$c : \type{Complex}$}
			\State $R \gets \set{~}$ \Comment{Initialize empty set of reactions}
			\State $c = (S, T), S = \set{s_1, s_2, \ldots s_{|S|}}$
			\ForAll{$i \in \set{1 \ldots |S|}$} \Comment{For each strand}
				\State $s_i = (z, D), D = \set{d_1, d_2, \ldots d_{|D|}}$
				\ForAll{$j \in \set{1 \ldots |D|}$} \Comment{For each \segment}
					\If{$T_{i,j} \ne \varnothing$ and $T_{i,j} > (i,j)$} \Comment{If \segment $d_j$ is paired to a later \segment}

						\Statex \Comment{Find the beginning and end of the helix where $d_j$ is bound}

						\State $A = (i_A, j_A) \gets (i,j)$ \Comment{Beginning of the helix on the ``top'' strand}
						\State $B = (i_B, j_B) \gets T_{i,j}$ \Comment{Beginning of the helix on the ``bottom'' strand}
						\State $A' = (i_{A'}, j_{A'}) \gets A$ \Comment{End of the helix on the ``top'' strand}
						\State $B' = (i_{B'}, j_{B'}) \gets B$ \Comment{End of the helix on the ``bottom'' strand}

						\State $\ell \gets |d_j|$ \Comment{Keep track of helix length in nucleotides}

						\State \Comment{Find the end of the helix}
						\While{$j_A' < |s_{i_{A'}}|$ and $j_B' \ge 0$ and $A' = T_{B'}$} \Comment{While neither strand is broken and \segments are paired}
							\State $j_{A'} \gets j_{A'} + 1$ \Comment{Move right}
							\State $j_{B'} \gets j_{B'} - 1$
							\State $\ell \gets \ell + |d_{j_{A'}}|$ 
						\EndWhile

						\State \Comment{Find the beginning of the helix}
						\While{$j_A \ge 0$ and $j_B < |s_{i_{B}}|$ and $A = T_{B}$} \Comment{While neither strand is broken and \segments are paired}
							\State $j_A \gets j_A - 1$ \Comment{Move left}
							\State $j_B \gets j_B + 1$
							\State $\ell \gets \ell + |d_{j_{A}}|$ 
						\EndWhile

						\If{$\ell < L$} \Comment{If the total length of the helix is less than the threshold}
							\State $T' \gets T$
							\ForAll {$j' \in \set{j_A \ldots j_{A'}}$}
								\State Update $T'$ to break pair $i, j'$
							\EndFor

							\If{$S, T'$ is not connected}
								\State $c', c'' \gets $ split $T'$ and $S$ into two connected complexes

								\State $R \gets R \cup \set{ ( \mset{c},\mset{c', c''} ) }$ \Comment{Make reaction that yields these complexes}
							\Else
								\State $c' \gets (S, T')$ \Comment{Make new complex}
								\State $R \gets R \cup \set{ ( \mset{c},\mset{c'} ) }$ \Comment{Make reaction that yields this complex}
							\EndIf
						\EndIf

					\EndIf
				\EndFor
			\EndFor
			\State \Return $R$

		\EndFunction
	\end{algorithmic}
	\end{algorithm}

	\begin{algorithm}[H]
	\caption{3-way branch migration move type}
	\label{alg:3way}
	\begin{algorithmic}[1]
		\Function{3way}{$c : \type{Complex}$}
			\State $R \gets \set{~}$ \Comment{Initialize empty set of reactions}
			\State $c = (S, T), S = \set{s_1, s_2, \ldots s_{|S|}}$
			\ForAll{$i \in \set{1 \ldots |S|}$} \Comment{For each strand}
				\State $s_i = (z, D), D = \set{d_1, d_2, \ldots d_{|D|}}$
				\ForAll{$j \in \set{1 \ldots |D|}$} \Comment{For each \segment}
					\State \Comment{Look to the left}
					\If{$T_{i,j} \ne \varnothing$ and $j \ne |D|$ and $T_{i,j+1} = \varnothing$}
						\State $(i', j') \gets (i, j+1)$ \Comment{$(i,j+1)$ will be the invading \segment} 

						\Repeat \Comment{Search \emph{left} to find a bound \segment $(i',j')$ that can be displaced}
							\State $j' \gets j'-1$
							\If{$T_{i', j'} \ne \varnothing$ and $d_{i', j'}$ can pair $d_{j+1}$}
								\State $T' \gets T$ where $T'_{i', j'} = (i, j+1)$, $T'_{i, j+1} = (i', j')$  

								\If{$S, T'$ is not connected}
									\State $c', c'' \gets $ split $T'$ and $S$ into two connected complexes

									\State $R \gets R \cup \set{ ( \mset{c},\mset{c', c''} ) }$ \Comment{Make reaction that yields these complexes}
								\Else
									\State $c' \gets (S, T')$ \Comment{Make new complex}
									\State $R \gets R \cup \set{ ( \mset{c},\mset{c'} ) }$ \Comment{Make reaction that yields this complex}
								\EndIf
							\EndIf
							\State $(i', j') \gets T_{i', j'}$ \Comment{Follow the structure}
						\Until{$j' = -1$ or $T_{i', j'} = \varnothing$ or $(i',j') = (i,j)$}
					\EndIf

					\State \Comment{Look to the right}
					\If{$T_{i,j} \ne \varnothing$ and $j \ne 0$ and $T_{i,j-1} = \varnothing$}
						\State $(i', j') \gets (i, j-1)$  \Comment{$(i,j-1)$ will be the invading \segment} 
						
						\Repeat \Comment{Search \emph{right} to find a bound \segment $(i',j')$ that can be displaced}
							\State $j' \gets j'+1$
							\If{$T_{i', j'} \ne \varnothing$ and $d_{i', j'}$ can pair $d_{j-1}$}
								\State $T' \gets T$ where $T'_{i', j'} = (i, j-1)$, $T'_{i, j-1} = (i', j')$

								\If{$S, T'$ is not connected}
									\State $c', c'' \gets $ split $T'$ and $S$ into two connected complexes

									\State $R \gets R \cup \set{ ( \mset{c},\mset{c', c''} ) }$ \Comment{Make reaction that yields these complexes}
								\Else
									\State $c' \gets (S, T')$ \Comment{Make new complex}
									\State $R \gets R \cup \set{ ( \mset{c},\mset{c'} ) }$ \Comment{Make reaction that yields this complex}
								\EndIf

							\EndIf
						\Until{$j' = |D|$ or $(i',j') = (i,j)$}
					\EndIf
				\EndFor
			\EndFor
			\State \Return $R$
		\EndFunction
	\end{algorithmic}
	\end{algorithm}

	\begin{algorithm}[H]
	\caption{4-way branch migration move type}
	\label{alg:4way}	
	\begin{algorithmic}[1]
		\Function{4way}{$c : \type{Complex}$}
			\State $R \gets \set{~}$ \Comment{Initialize empty set of reactions}
			\State $c = (S, T), S = \set{s_1, s_2, \ldots s_{|S|}}$
			\ForAll{$i \in \set{1 \ldots |S|}$} \Comment{For each strand}
				\State $s_i = (z, D), D = \set{d_1, d_2, \ldots d_{|D|}}$
				\ForAll{$j \in \set{1 \ldots |D|}$} \Comment{For each \segment}
					
					\If{$T_{i,j} \ne \varnothing$ and $j < |D|$} \Comment{\segment $d_j$ must be bound, not be at the end of a strand}

						\State $A \gets (i, j+1)$  \Comment{Displacing \segment}
						\State $B \gets T_{i, j+1}$ \Comment{Displaced \segment}
						\If{$B = \varnothing$} Continue \EndIf

						\State $C = (i_C, j_C) \gets T_{i, j}$
						\State $C \gets (i_C, j_C-1)$ \Comment{Template \segment (replaces B, binds A)}
						\If{$j_C < 1$} Continue \EndIf

						\State $D \gets T_{C}$ \Comment{(replaces A, binds B)}
						\If{$D = \varnothing$} Continue \EndIf

						\If{$B \ne C$ and $d_A = d_C^*$ and $d_B = d_D^*$} \Comment{If this is a four-way branch migration}
							\State $T' \gets T$, $T'_A \gets C, T'_C \gets A$; $T'_B \gets D, T'_D \gets B$

							\If{$S, T'$ is not connected}
								\State $c', c'' \gets $ split $T'$ and $S$ into two connected complexes

								\State $R \gets R \cup \set{ ( \mset{c},\mset{c', c''} ) }$ \Comment{Make reaction that yields these complexes}
							\Else
								\State $c' \gets (S, T')$ \Comment{Make new complex}
								\State $R \gets R \cup \set{ ( \mset{c},\mset{c'} ) }$ \Comment{Make reaction that yields this complex}
							\EndIf

						\EndIf

					\EndIf

				\EndFor
			\EndFor
			\State \Return $R$
		\EndFunction
	\end{algorithmic}
	\end{algorithm}

	\begin{algorithm}[H]
	\caption{Reaction enumeration}
	\begin{algorithmic}[1]
	\Procedure{Enumerate}{$A : \set{\type{Complex}}$}
			\State $\li{E} \gets \set{}; \li{S} \gets \set{}; \li{T} \gets \set{}$      \Comment{Complexes}
			\State $\li{R} \gets \set{}$                                                \Comment{Reactions}
			\State $\li{Q} \gets \set{}$                                                \Comment{Resting states}
			\State $\li{B} \gets A$
			\While{ $\li{B} \ne \set{}$ }                                               \Comment{Enumerate fast reactions from $A$}
					\State $b \gets$ pop($\li{B}$)
					\State $(S', T', Q', R') \gets$ \Call{enumerateNeighborhood}{$b$}   \Comment{Find fast reactions from $b$}
					\State $\li{S} \gets \li{S} \cup S'$; $\li{T} \gets \li{T} \cup T'$; $\li{R} \gets \li{R} \cup R'$; $\li{Q} \gets \li{Q} \cup Q'$

			\EndWhile
			
			\While{$\li{S} \ne \set{}$}                                                 \Comment{Enumerate slow reactions between resting state complexes}
					\State $s \gets$ pop($\li{S}$)
					\State $(R', B') \gets$ \Call{getSlowReactions}{$s$, $\li{S} \cup \li{E}$}      \Comment{Find slow reactions from $s$}
					\State $\li{E} \gets \li{E} \cup \set{s}$                               \Comment{$s$ moves to $\li{E}$ once slow reactions are enumerated}
					\State $\li{R} \gets \li{R} \cup R'$                                    \Comment{Store new reactions}
					\State $\li{B} \gets \li{B} \cup B' \setminus (\li{E} \cup \li{S} \cup {T})$ \Comment{Store new complexes}

					\While{ $\li{B} \ne \set{}$ }                                           \Comment{Enumerate fast reactions from $B$}
							\State $b \gets$ pop($\li{B}$)
							\State $(S', T', Q', R') \gets$ \Call{neighborhood}{$b$}   \Comment{Find fast reactions from $b$}
							\State $\li{S} \gets \li{S} \cup S'$; $\li{T} \gets \li{T} \cup T'$; $\li{R} \gets \li{R} \cup R'$; $\li{Q} \gets \li{Q} \cup Q'$
					\EndWhile

			\EndWhile
	\EndProcedure
	\newline

	\Procedure{enumerateNeighborhood}{$c : \type{Complex}$}                     \Comment{Calculates fast reactions from $c$, sorts complexes into resting state complexes/transient complexes}
			\State $\li{F} = \{ c \}$                                           \Comment{Complexes from fast reactions in neighborhood}
			\State $\li{N} = \set{ }$                                           \Comment{Complexes in neighborhood}
			\State $\li{R}_N = \set{ }$                                         \Comment{(Fast) Reactions in neighborhood}
			\While{ $\li{F} \ne \set{}$ }                                       \Comment{Enumerate fast reactions from each complex in $F$}
					\State $f \gets$ pop($\li{F}$)
					\State $(R', F') \gets$ \Call{getFastReactions}{$f$}        \Comment{Find fast reactions from $f$}
					\State $\li{F}   \gets \li{F} \cup F' \setminus \li{N}$
					\State $\li{N}   \gets \li{N} \cup F'$
					\State $\li{R}_N \gets \li{R}_N \cup R'$
			\EndWhile

			\State Apply Tarjan's algorithm\cite{Tarjan:1972hk} to find strongly-connected components of the directed graph $G = (\li{N},\li{R}_N)$
			\State $Q' \gets \set{\text{strongly-connected components of } G \text{ with no outgoing fast reactions}}$             \Comment{Resting states are SCCs of $G$}
			\State $S' \gets \set{ s : s \in q ~\text{for any}~ q \in Q'}$      \Comment{resting state complexes are in a resting state}
			\State $T' \gets \li{N} \setminus S'$                               \Comment{Transient complexes are everything else}

			\State \Return $(S', T', Q', \li{R}_N)$
	\EndProcedure
	\newline

	\Procedure{getFastReactions}{$c : \type{Complex}$}
		\Statex \Comment{Calculates all fast (unimolecular) reactions that consume $c$}
		\State $R \gets$ fast reactions consuming $c$, $C \gets$ union of products of reactions in $R$
		\State \Return $R, C$
	\EndProcedure

	\Procedure{getSlowReactions}{$c : \type{Complex}$, $S : \set{\type{Complex}}$}
		\Statex \Comment{Calculates all slow (bimolecular) reactions that consume $c$ and an element of $S$}
		\State $R \gets$ slow reactions consuming $c$ and $s \in S$, $C \gets$ union of products of reactions in $R$
		\State \Return $R, C$
	\EndProcedure

	\end{algorithmic}
	\label{alg:rxn-enum}
	\end{algorithm}

	\begin{algorithm}[H]
	\caption{Condensing Reactions}
	\begin{algorithmic}[1]

	\State $\FT{x} \gets$ undefined $\forall$ complexes $x$              \Comment{The map $\FT : \type{Complex} \to \set{\type{Fate}}$ is global and begins empty}
	\State $S \gets \set{~}$              \Comment{The map $S : \type{Complex} \to \set{\type{Complex}}$ is global and begins empty}
	\Statex

	\Procedure{Condense}{$\li{C} : \set{\type{Complex}}$, $\li{R} : \set{\type{Reaction}}$}     \Comment{Computes the fates for each complex, then generates a set of condensed reactions}
			\State $\li{R}_s \gets \set{r : r \in \li{R}, r~\text{is slow} }$                       \Comment{Slow reactions}
			\State $\li{R}_f \gets \li{R} \setminus \li{R}_s$                                       \Comment{Fast reactions}
			\State $\li{R}_f^{(1,1)} \gets \set{r \in \li{R}_f : \alpha(r) = (1,1) }$				\Comment{Fast (1,1) reactions}
			\State Use Tarjan's algorithm\cite{Tarjan:1972hk} to compute the set of strongly-connected components 
			\Statex from the graph $\Gamma = (\li{C}, \li{R}_f^{(1,1)})$
			\State $S \gets$ the set of strongly-connected components of $\Gamma$
			\State $\SC(x) \gets$ the strongly-connected component containing complex $x$, $\forall x \in \li{C}$

			\ForAll{$\li{C}_c \in S$}                                                               \Comment{For each SCC $\li{C}_c$ of $\Gamma$}
					\State \Call{computeFates}{$C_c$, $\li{R}_f$}
			\EndFor
			\State \Return \Call{condenseReactions}{$\li{R}_s$}
	\EndProcedure
	\Statex

	\Procedure{computeFates}{$\li{C} : \set{\type{Complex}}$, $\li{R}_f : \set{\type{Reaction}}$}

			\State $R_o \gets \set{r = (A, B) \in \li{R}_f : A \subseteq \li{C}, B \setminus \li{C} \neq \emptyset }$               \Comment{Outgoing fast reactions}		
			\If{$|R_o| = 0$}                                                                    \Comment{If no outgoing fast reactions}
					\State $\FT(c) \gets \mset{\li{C}}~\forall~c \in \li{C}$                    \Comment{$\FT(c)$ is the resting state $\li{C}$ containing the complex $c$}
			\Else                                                                               \Comment{If there are outgoing fast reactions}
					\ForAll{$c \in \li{C}$}
							\State $R_o^(1,n) \gets \set{r \in R_o : \alpha(r) = (1, n)}$
							\State $P_o \gets \bigcup_{r = (A, B) \in R_o^(1,1)} B$ 
							\ForAll{$x \in P_o$}
									\State If $\FT(x)$ is undefined, \Call{computeFates}{$\SC(x)$, $\li{R}_f$}
							\EndFor
							\State $\FT(c) \gets \displaystyle\bigcup_{r = (A,B) \in R_o} \left( \bigoplus_{b \in B} \FT(b) \right)$  \Comment{$\FT(c)$ are the possible fates from outgoing reactions}
					\EndFor
			\EndIf 

	\EndProcedure
	\Statex

	\Procedure{condenseReactions}{$\li{R}_s : \set{\type{Reaction}}$} \Comment{Condensed reaction space from the set of slow reactions}
			\State $\li{\hat{R}} \gets \set{~}$                     \Comment{Condensed reactions}
			\ForAll{$s = (A, {b}) \in \li{R}_s$}
					\State $A'   \gets \sum_{a \in A} \FT(a)$          \Comment{Fates of reactants are all trivial}
					\ForAll{$B' \in \FT(b)$}							\Comment{For each fate of $b$}
						\State $r' \gets (A', B')$						\Comment{Generate new reaction}
						\State $\li{\hat{R}} \gets \li{\hat{R}} \cup {r'}$ 
					\EndFor
			\EndFor
			\State \Return $\li{\hat{R}}$
	\EndProcedure

	\end{algorithmic}
	\label{alg:rxn-condense}
	\end{algorithm}

\end{document}